\documentclass{siamltex1213}
\usepackage[english]{babel}
\usepackage{amsmath}
\usepackage{amstext}
\usepackage{amssymb}
\usepackage{mathtools}

\def\R{\mathbb{R}}

\def\rP{\mathbb{P}}

\def\Lip{\mathop{\rm Lip}}
\def\Gr{\mathop{\rm Gr}}

\def\hx{\hat{x}}

\def\C{{\mathcal C}}

\def\P{{\mathcal P}}

\def\M{{\mathcal M}}

\def\bmu{{\boldsymbol \mu}}
\def\bnu{{\boldsymbol \nu}}
\def\bpi{{\boldsymbol \pi}}
\def\bxi{{\boldsymbol \xi}}
\def\by{{\bf y}}
\def\ba{{\bf a}}
\def\bx{{\bf x}}

\def\bu{{\bf u}}
\def\bv{{\bf v}}

\def\tx{{\tilde x}}
\def\te{{\tilde e}}
\def\tpi{{\tilde \pi}}
\def\hx{{\hat x}}

\def\sX{{\mathsf X}}

\def\sA{{\mathsf A}}
\def\sH{{\mathsf H}}

\def\sM{{\mathsf M}}
\def\sE{{\mathsf E}}

\def\sG{{\mathsf G}}
\def\sW{{\mathsf W}}

\newtheorem{remark}{Remark}

\allowdisplaybreaks

\newcommand{\appsec}{
\renewcommand{\thesubsection}{\Alph{subsection}}
}

\begin{document}
\sloppy
\title{Markov-Nash Equilibria in Mean-Field Games with Discounted Cost
\thanks{This research was supported in part by the U.S.\ Air Force Office of Scientific Research (AFOSR) under MURI grant FA9550-10-1-0573, and in part by the Office of Naval Research under (ONR) MURI grant N00014-16-1-2710 and grant N00014-12-1-0998.}
}
\author{Naci Saldi, Tamer Ba\c{s}ar, and Maxim Raginsky
\thanks{The authors are with the Coordinated Science Laboratory, University of Illinois at Urbana-Champaign, Urbana, IL 61801, Email: \{nsaldi,basar1,maxim\}@illinois.edu}
     }
\maketitle

\begin{abstract}
In this paper, we consider discrete-time dynamic games of the mean-field type with a finite number $N$ of agents subject to an infinite-horizon discounted-cost optimality criterion. The state space of each agent is a locally compact Polish space.  At each time, the agents are coupled through the empirical distribution of their states, which affects both the agents' individual costs and their state transition probabilities. We introduce a new solution concept of the Markov-Nash equilibrium, under which a policy is player-by-player optimal in the class of all Markov policies. Under mild assumptions, we demonstrate the existence of a mean-field equilibrium in the infinite-population limit $N \to \infty$, and then show that the policy obtained from the mean-field equilibrium is approximately Markov-Nash when the number of agents $N$ is sufficiently large.
\end{abstract}
\begin{keywords}
Mean-field games, Nash equilibrium, discounted cost. 
\end{keywords}

\begin{AMS}
91A15, 91A10, 91A13, 93E20.  	    	
\end{AMS}

\section{Introduction}\label{sec1}

Mean-field game theory has been developed independently by Huang, Malham\'{e}, and Caines \cite{HuMaCa06} and Lasry and Lions \cite{LaLi07} to study continuous-time non-cooperative differential games with a large number of identical agents.  The key underlying idea is that, under a particular equilibrium condition, called the Nash certainty equivalence principle \cite{HuMaCa06}, the decentralized game problem can be reduced to a single-agent decision problem. This principle essentially says that the state evolution of an individual agent should be consistent with the total population behavior.

The optimal solution of this single-agent decision problem provides an approximation to Nash equilibria of games with large (but finite) population sizes. However, in contrast to the standard single-agent optimal control problem, the characterization of this optimal solution in the continuous-time setting leads to a Fokker-Planck equation evolving forward in time, and a Hamilton-Jacobi-Bellman equation evolving backward in time. We refer the reader to \cite{HuCaMa07,TeZhBa14,Hua10,BeFrPh13,Ca11,CaDe13,GoSa14,BaTeBa16,MoBa17,MoBa16} for studies of continuous-time mean-field games with different models and cost functions, such as games with major-minor players, risk-sensitive games, robust mean field games, games with jump parameters, and LQG games.

By contrast, there are relatively few results on \textit{discrete-time} mean-field games. Existing works have mostly studied the setup where the state space is discrete (finite or countable) and the agents are coupled only through their cost functions; that is, the mean-field term does not affect the evolution of the states of the agents. Gomes et al.~\cite{GoMoSo10} study a discrete-time mean-field game with a finite state space over a finite horizon. Adlakha et al.~\cite{AdJoWe15} consider a discrete-time mean-field game where the state space is a countable subset of a $d$-dimensional lattice, under an infinite-horizon discounted cost criterion. This setup is the closest to the one studied here. However, in addition to the state space being at most countable, Adlakha et al.~\cite{AdJoWe15} also assume that, for each agent, every stationary policy induces a Markov chain that has at least one invariant distribution. In this case, the optimal control problem in the mean-field limit corresponds to a standard homogeneous Markov decision problem. Biswas \cite{Bis15} considers the average-cost setting, where the state space is a $\sigma$-compact Polish space and the transition probability of an individual agent does not depend on the empirical distribution of the states. Under the regularity conditions imposed in \cite{Bis15}, it can be shown that, for each stationary policy, the corresponding Markov process for each agent has a unique invariant measure, which is a standard technique when studying average-cost problems. However, these regularity conditions are stated in terms of a specific metric topology on the set of stationary policies, and appear to be too strong to hold under reasonable assumptions. Discrete-time mean-field games with linear individual dynamics are studied in \cite{ElLiNi13,MoBa15,NoNa13,MoBa16-cdc}.

In this paper, we consider discrete-time mean-field games with a Polish state space, under the infinite-horizon discounted-cost optimality criterion. In such a game, the agents are coupled through the empirical distribution of their states at each time, which affects both the individual costs and the state transition probabilities of each agent. In Section~\ref{sec2} we formulate the finite-agent discrete-time game problem of the mean-field type and introduce the solution concept of \textit{Markov-Nash equilibrium}, under which a policy is player-by-player optimal in the class of all Markov policies. In Section~\ref{sec3},  we introduce the infinite-population mean-field game and prove the existence of a mean-field equilibrium, which we use in Section~\ref{sec4} to deduce the existence of approximate Markov-Nash equilibria for $N$-agent games when $N$ is sufficiently large. In Section~\ref{example} we present an example to illustrate our results. Section~\ref{conc} lists some directions for future research. To the best of our knowledge, this is the first result demonstrating the existence of an (exact or approximate) equilibrium policy for a general discrete-time mean-field game with finitely many agents.

\smallskip

\noindent\textbf{Notation.} For a metric space $\sE$, we let $C_b(\sE)$ denote the set of all bounded continuous real functions on $\sE$. Let $\P(\sE)$ denote the set of all Borel probability measures on $\sE$. For any $\sE$-valued random element $x$, ${\cal L}(x)(\,\cdot\,) \in \P(\sE)$ denotes the distribution of $x$. A sequence $\{\mu_n\}$ of measures on $\sE$ is said to converge weakly to a measure $\mu$ if $\int_{\sE} g(e) \mu_n(de)\rightarrow\int_{\sE} g(e) \mu(de)$ for all $g \in C_b(\sE)$. For any $\nu \in \P(\sE)$ and measurable real function $g$ on $\sE$, we define $\nu(g) \coloneqq \int g d\nu$. For any subset $B$ of $\sE$, we let $\partial B$ and $B^c$ denote the boundary and complement of $B$, respectively. The notation $v\sim \nu$ means that the random element $v$ has distribution $\nu$. For real numbers $a$ and $b$, the notation $a\vee b$ denotes the maximum of $a$ and $b$. Unless otherwise specified, the term ``measurable" will refer to Borel measurability.


\section{Finite Player Game with Discounted Cost}\label{sec2}

We consider a discrete-time $N$-agent stochastic game with a Polish state space $\sX
$ and a Polish action space $\sA$. The dynamics of the game are specified by an initial state distribution $\mu_0 \in \P(\sX)$ and a stochastic state transition kernel $p : \sX \times \sA \times \P(\sX) \to \P(\sX)$ as follows. For every $t \in \{0,1,2,\ldots\}$ and every $i \in \{1,2,\ldots,N\}$, let $x^N_i(t) \in \sX$ and $a^N_i(t) \in \sA$ denote the state and the action of Agent~$i$ at time $t$, and let
\begin{align}
e_t^{(N)}(\,\cdot\,) \coloneqq \frac{1}{N} \sum_{i=1}^N \delta_{x_i^N(t)}(\,\cdot\,) \in \P(\sX) \nonumber
\end{align}
denote the empirical distribution of the state configuration at time $t$, where $\delta_x\in\P(\sX)$ is the Dirac measure at $x$. The initial states $x^N_i(0)$ are independent and identically distributed according to $\mu_0$, and, for each $t \ge 0$, the next-state configuration $(x^N_1(t+1),\ldots,x^N_N(t+1))$ is generated at random according to the probability law
\begin{align}\label{eq:state_spec}
\prod^N_{i=1} p\big(dx^N_i(t+1)\big|x^N_i(t),a^N_i(t),e^{(N)}_t\big),
\end{align}
where $p(\,\cdot\,|x,a,\mu)$ denotes the image of the triple $(x,a,\mu) \in \sX \times \sA \times \P(\sX)$ in $\P(\sX)$ under the state transition kernel $p$.

To complete the description of the game dynamics, we must specify how the agents select their actions at each time step. To that end, we introduce the history spaces $\sH_0 = \sX \times \P(\sX)$ and $\sH_{t}=(\sX\times\sA\times\P(\sX))^{t}\times (\sX\times\P(\sX))$ for $t=1,2,\ldots$, all endowed with product Borel $\sigma$-algebras.\footnote{We endow the set $\P(\sX)$ with the topology of weak convergence, which makes it a Polish space.} A \emph{policy} for a generic agent is a sequence $\pi=\{\pi_{t}\}$ of stochastic kernels on $\sA$ given $\sH_{t}$; we say that such a policy is \emph{Markov} if each $\pi_t$ is a Markov kernel on $\sA$ given $\sX$. The set of all policies for Agent~$i$ is denoted by $\Pi_i$, and the subset consisting of all Markov policies by $\sM_i$. Furthermore, we let $\sM_i^c$ denote the set of all Markov policies for Agent~$i$ that are weakly continuous; that is, $\pi=\{\pi_t\}\in\sM_i^c$ if for all $t\geq0$, $\pi_t: \sX \rightarrow \P(\sA)$ is continuous when $\P(\sA)$ is endowed with the weak topology.

Let ${\bf \Pi}^{(N)} = \prod_{i=1}^N \Pi_i$, ${\bf \sM}^{(N)} = \prod_{i=1}^N \sM_i$, and ${\bf \sM}^{(N,c)} = \prod_{i=1}^N \sM_i^c$.
We let ${\boldsymbol \pi}^{(N)} \coloneqq (\pi^1,\ldots,\pi^N)$, $\pi^i \in \Pi_i$ denote the $N$-tuple of policies for all the agents in the game. We will refer to ${\boldsymbol \pi}^{(N)} \in {\bf \Pi}^{(N)}$ simply as a `policy.' Under such a policy, the action configuration at each time $t \ge 0$ is generated at random according to the probability law
\begin{align}\label{eq:policy_spec}
\prod^N_{i=1} \pi^i_t\big(da^N_i(t)\big|h^N_i(t)\big),
\end{align}
where $h^N_i(0) = (x^N_i(0),e^{(N)}_0)$ and $h^N_i(t) = (h^N_i(t-1),x^N_i(t),a^N_i(t-1),e^{(N)}_t)$ for $t \ge 1$ are the histories observed by Agent~$i$ at each time step. When ${\boldsymbol \pi}^{(N)} \in {\bf \sM}^{(N)}$, Eq.~\eqref{eq:policy_spec} becomes
$$
\prod^N_{i=1}\pi^i_t(da^N_i(t)|x^N_i(t)).
$$
The stochastic update rules in Eqs.~\eqref{eq:state_spec} and \eqref{eq:policy_spec}, together with the initial state distribution $\mu_0$, uniquely determine the probability law of all the states and actions for all $i \in \{1,\ldots,N\}$ and all $t \geq 0$. We will denote expectations with respect to this probability law by $E^{{\boldsymbol \pi}^{(N)}}\big[\cdot\big]$.

We now turn to the question of optimality. The \emph{one-stage cost} function for a generic agent is a measurable function $c : \sX \times \sA \times \P(\sX) \to [0,\infty)$. For Agent~$i$, the infinite-horizon discounted cost under the initial distribution $\mu_0$ and a policy ${\boldsymbol \pi}^{(N)} \in {\bf \Pi}^{(N)}$ is given by
\begin{align}
J_i^{(N)}({\boldsymbol \pi}^{(N)}) &= E^{{\boldsymbol \pi}^{(N)}}\biggl[\sum_{t=0}^{\infty}\beta^{t}c(x_{i}^N(t),a_{i}^N(t),e^{(N)}_t)\biggr], \nonumber
\end{align}
where $\beta \in (0,1)$ is the discount factor. The standard notion of optimality is a player-by-player one:

\begin{definition}
A policy ${\boldsymbol \pi}^{(N*)}= (\pi^{1*},\ldots,\pi^{N*})$ constitutes a \emph{Nash equilibrium} if
\begin{align}
J_i^{(N)}({\boldsymbol \pi}^{(N*)}) = \inf_{\pi^i \in \Pi_i} J_i^{(N)}({\boldsymbol \pi}^{(N*)}_{-i},\pi^i) \nonumber
\end{align}
for each $i=1,\ldots,N$, where ${\boldsymbol \pi}^{(N*)}_{-i} \coloneqq (\pi^{j*})_{j\neq i}$.
\end{definition}

There are two challenges pertaining to Nash equilibria in the class of games considered here. The first challenge is the (almost) decentralized nature of the information structure of the problem. Namely, the agents have access only to their local state information $x_i^N(t)$ and the empirical distribution of the states $e_t^{(N)}$, and there is no general theory (of existence and characterization of Nash equilibria) for such games even with special structures for the transition probabilities. The second difficulty arises because of the so-called \emph{curse of dimensionality}; that is, the solution of the problem becomes intractable when the number of states/actions and agents is large. The existence of Nash equilibria in this case is a challenging problem even when the agents have access to full state information $\{x_i^N(t)\}_{i=1}^N$ (see \cite{JaNo16,BaDu14,LeMc15} and references therein). Therefore, it is of interest to find an approximate decentralized equilibrium with reduced complexity. To that end, we adopt the following solution concept:
\begin{definition}\label{def1}
A policy ${\boldsymbol \pi}^{(N*)} \in {\bf \sM}^{(N)}$ is a \emph{Markov-Nash equilibrium} if
\begin{align*}
J_i^{(N)}({\boldsymbol \pi}^{(N*)}) &= \inf_{\pi^i \in \sM_i} J_i^{(N)}({\boldsymbol \pi}^{(N*)}_{-i},\pi^i)
\end{align*}
for each $i=1,\ldots,N$, and an \emph{$\varepsilon$-Markov-Nash equilibrium} (for a given $\varepsilon > 0$) if
\begin{align*}
J_i^{(N)}({\boldsymbol \pi}^{(N*)}) &\leq \inf_{\pi^i \in \sM_i} J_i^{(N)}({\boldsymbol \pi}^{(N*)}_{-i},\pi^i) + \varepsilon
\end{align*}
for each $i=1,\ldots,N$.
\end{definition}

The main contribution of this paper is the proof of existence of $\varepsilon$-Markov-Nash equilibria for games with sufficiently many agents. To this end, we first consider a mean-field game that arises in the infinite-population limit $N \to \infty$ and prove the existence of an appropriately defined mean-field equilibrium for this limiting mean-field game. Then we pass back to the finite-$N$ case and show that, if each agent in the original problem adopts the mean-field eqiulibrium policy, then the resulting policy will be an approximate Markov-Nash equilibrium for all sufficiently large $N$. It is important to note that, although the policy in the mean-field equilibrium is an approximate Markov-Nash equilibrium for the finite-agent game problem, it is indeed a true Nash equilibrium in the infinite population regime. This follows from the fact that the set of Markov policies is sufficiently rich for optimality in the limiting case, as each agent is faced with a single-agent decision problem.

\subsection{Assumptions}\label{assumptions}

In this section, we state all assumptions that will be made on the game model for easy reference. They will be imposed when needed in the remainder of the paper.

Let $w: \sX \rightarrow [1,\infty)$ be a continuous moment function; that is, there exists an increasing sequence of compact subsets $\{K_n\}_{n\geq1}$ of $\sX$ such that
\begin{align}
\lim_{n\rightarrow\infty} \inf_{x \in \sX \setminus K_n} w(x) = \infty. \nonumber
\end{align}
We assume that $w(x) \geq 1 + d_{\sX}(x,x_0)^p$ for some $p\geq1$ and $x_0 \in \sX$, where $d_{\sX}$ is the metric on $\sX$. In order to study bounded and unbounded one-stage cost functions $c$ simultaneously, we define the following function:
\begin{align}
v \coloneqq \begin{cases} 1, & \text{if $c$ is bounded} \\
w, & \text{if $c$ is unbounded}
\end{cases}\nonumber
\end{align}
For any $g: \sX \rightarrow \R$, define the $v$-norm of $g$ as
\begin{align}
\|g\|_v \coloneqq \sup_{x \in \sX} \frac{|g(x)|}{v(x)}. \nonumber
\end{align}
Let $B_v(\sX)$ denote the Banach space of all real valued measurable functions $g$ on $\sX$ with finite $v$-norm and let $C_v(\sX)$ denote the Banach space of all real valued continuous functions in $B_v(\sX)$.

Analogously, for any finite signed measure $\mu$ on $\sX$, let us define the $v$-norm of $\mu$ as
\begin{align}
\|\mu\|_v \coloneqq \sup_{\substack{g \in B_v(\sX): \\ \|g\|_v \leq 1}} \biggl| \int_{\sX} g(x) \mu(dx) \biggr|. \nonumber
\end{align}
Note that if $v = 1$, then $\|\mu\|_v$ corresponds to the total variation distance \cite[Section 7.2]{HeLa99}.
Let
\begin{align}
\P_v(\sX) &\coloneqq \big\{ \mu \in \P(\sX): \|\mu\|_v < \infty \bigr\} \nonumber \\
&= \biggl\{ \mu \in \P(\sX): \int_{\sX} v(x) \mu(dx) < \infty \biggr\}. \nonumber
\end{align}
It is known that weak topology on $\P(\sX)$ can be metrized using the following metric:
\begin{align}
\rho(\mu,\nu) \coloneqq \sum_{m=1}^{\infty} 2^{-(m+1)} | \mu(f_m) - \nu(f_m) |, \nonumber
\end{align}
where $\{f_m\}_{m\geq1}$ is an appropriate sequence of continuous and bounded functions such that $\|f_m\| \leq 1$ for all $m\geq1$ (see \cite[Theorem 6.6, p. 47]{Par67}). We define
\begin{align}
\rho_v(\mu,\nu) \coloneqq \rho(\mu,\nu) + |\mu(v) - \nu(v)| \nonumber
\end{align}
which is a metric on $\P_v(\sX)$. Note that
\begin{align}
\rho_v(\mu_n,\mu) \rightarrow 0 \Longleftrightarrow \mu_n(g) \rightarrow \mu(g), \text{ } \forall g \in C_v(\sX). \nonumber
\end{align}
Indeed, ($\Leftarrow$) is trivial. For ($\Rightarrow$), let $\rho_v(\mu_n,\mu) \rightarrow 0$ which means that $\mu_n \rightarrow \mu$ weakly and $\mu_n(v) \rightarrow \mu(v)$. Let $g \in C_v(\sX)$. Define non-negative continuous function $h$ as $h \coloneqq \|g\|_v v + g$. It is known that \cite[Proposition E.2]{HeLa96}
\begin{align}
\liminf_{n\rightarrow\infty} \int_{\sX} h(x) \mu_n(dx) \geq \int_{\sX} h(x) \mu(dx). \nonumber
\end{align}
But, since $\mu_n(v) \rightarrow \mu(v)$, we should have
\begin{align}
\liminf_{n\rightarrow\infty} \int_{\sX} g(x) \mu_n(dx) \geq \int_{\sX} g(x) \mu(dx). \nonumber
\end{align}
Conversely, define non-negative continuous function $u$ as $u \coloneqq \|g\|_v v - g$. Then, we have
\begin{align}
\liminf_{n\rightarrow\infty} \int_{\sX} u(x) \mu_n(dx) \geq \int_{\sX} u(x) \mu(dx). \nonumber
\end{align}
But, since $\mu_n(v) \rightarrow \mu(v)$, we should also have
\begin{align}
\limsup_{n\rightarrow\infty} \int_{\sX} g(x) \mu_n(dx) \leq \int_{\sX} g(x) \mu(dx). \nonumber
\end{align}
Thus, $\mu_n(g) \rightarrow \mu(g)$, which establishes the result. Let us call the topology induced by metric $\rho_v$ on $\P_v(\sX)$ as $v$-topology. It can be proved that $\P_v(\sX)$ with metric $\rho_v$ is a Polish space.

Suppose that $v=w$ (i.e., $c$ is unbounded). Define the Wasserstein distance of order $p\geq1$ on $\P_v(\sX)$ as follows \cite[Definition 6.1]{Vil09}:
\begin{align}
W_p(\mu,\nu) \coloneqq \inf \bigl\{ E[d_{\sX}(X,Y)^p]^{\frac{1}{p}}: {\cal L}(X) = \mu \text{ and } {\cal L}(Y) = \nu \bigr\}. \nonumber
\end{align}
Note that $W_p$ is a metric on $\P_v(\sX)$ since $v(x) \geq 1 + d_{\sX}(x,x_0)^p$ for all $x\in\sX$. Furthermore, $W_p(\mu_n,\mu) \rightarrow 0$ if and only if $\mu_n(g) \rightarrow \mu(g)$ for all continuous $g$ with $|g(x)| \leq 1+ d_{\sX}(x,x_0)^p$ \cite[Definition 6.8, Theorem 6.9]{Vil09}. The last observation implies that the following metric $\beta_v(\mu,\nu) \coloneqq W_p(\mu,\nu) + |\mu(v)-\nu(v)|$ metrizes the $v$-topology on $\P_v(\sX)$. Using the dual formulation of $W_p$ \cite[Theorem 5.10]{Vil09}, we can write $\beta_v$ as follows:
\begin{align}
\beta_v(\mu,\nu) = \sup_{\substack{(h,g) \in {\cal L}_1(\mu)\times {\cal L}_1(\nu):\\ h(x)-g(y) \leq d_{\sX}(x,y)^p}} |\mu(h) - \nu(g)|^{\frac{1}{p}} + |\mu(v)-\nu(v)|, \label{beta}
\end{align}
where ${\cal L}_1(\lambda)$ denotes the set of all $\lambda$-integrable real functions on $\sX$.

\begin{remark}
In the remainder of the paper, $\P(\sX)$ is always equipped with the weak topology while $\P_v(\sX)$ is always equipped with the $v$-topology. In other words, when we say that a function over $\P_v(\sX)$ is continuous, it should be understood that it is continuous with respect to $v$-topology. Similarly, a function over $\P(\sX)$ is continuous if it is continuous with respect to weak topology.
\end{remark}

\smallskip

\textbf{Assumption~1:}
\begin{itemize}
\item [(a)] The one-stage cost function $c$ is continuous.
\item [(b)] $\sA$ is compact and $\sX$ is locally compact.
\item [(c)] There exists a non-negative real number $\alpha$ such that
\begin{align}
\sup_{(a,\mu) \in \sA \times \P(\sX)} \int_{\sX} w(y) p(dy|x,a,\mu) \leq \alpha w(x). \nonumber
\end{align}
\end{itemize}

\begin{itemize}
\item [(d)] The stochastic kernel $p(\,\cdot\,|x,a,\mu)$ is weakly continuous; that is, if $(x_n,a_n,\mu_n) \rightarrow (x,a,\mu)$ in $\sX \times \sA \times \P(\sX)$, then $p(\,\cdot\,|x_n,a_n,\mu_n) \rightarrow p(\,\cdot\,|x,a,\mu)$ weakly. In addition, the function $\int_{\sX} w(y) p(dy|x,a,\mu)$ is continuous in $(x,a,\mu)$.
\item [(e)] The initial probability measure $\mu_0$ satisfies
\begin{align}
\int_{\sX} v(x) \mu_0(dx) \eqqcolon M < \infty. \nonumber
\end{align}
\end{itemize}

\begin{remark}\label{remark1}
Note that Assumption~1-(c) implies that the range of $p$ lies in $\P_v(\sX)$; that is, $\bigl\{p(\,\cdot\,|x,a,\mu): (x,a,\mu) \in \sX \times \sA \times \P(\sX) \bigr\} \subset \P_v(\sX)$. Therefore, Assumption~1-(d) is equivalent to the following condition: if $(x_n,a_n,\mu_n) \rightarrow (x,a,\mu)$ in $\sX \times \sA \times \P(\sX)$, then $p(\,\cdot\,|x_n,a_n,\mu_n) \rightarrow p(\,\cdot\,|x,a,\mu)$ with respect to the $v$-topology on $\P_v(\sX)$.
\end{remark}

\smallskip

For each $t\geq0$, let us define
\begin{align}
\P_v^t(\sX) \coloneqq \biggl\{ \mu \in \P_v(\sX): \int_{\sX} w(x) \mu(dx) \leq \alpha^t M \biggr\}. \nonumber
\end{align}

\begin{itemize}
\item [(f)] There exist $\gamma \geq1$ and a non-negative real number $R$ such that for each $t\geq0$, if we define $M_t \coloneqq \gamma^t R$, then
\begin{align}
\sup_{(a,\mu) \in \sA \times \P_v^t(\sX)} c(x,a,\mu) \leq M_t v(x). \nonumber
\end{align}
\item [(g)] We assume that $\alpha \beta \gamma < 1$.
\end{itemize}
\vspace{5pt}

\textbf{Assumption~2:}
\vspace{5pt}

Define the following moduli of continuity:
\begin{align}
\omega_{p}(r) &\coloneqq \sup_{(x,a) \in \sX\times\sA} \sup_{\substack{\mu,\nu: \\ \tilde{\rho}_v(\mu,\nu)\leq r}} \|p(\,\cdot\,|x,a,\mu) - p(\,\cdot\,|x,a,\nu)\|_{v} \nonumber \\
\omega_{c}(r) &\coloneqq \sup_{(x,a) \in \sX\times\sA} \sup_{\substack{\mu,\nu: \\ \tilde{\rho}_v(\mu,\nu)\leq r}} |c(x,a,\mu) - c(x,a,\nu)|, \nonumber
\end{align}
where $\tilde{\rho}_v = \beta_v$ (see (\ref{beta})) if $c$ is unbounded, and $\tilde{\rho}_v=\rho$ if $c$ is bounded.

For any function $g: \P_v(\sX) \rightarrow \R$, we define the $v$-norm of $g$ as follows:
\begin{align}
\|g\|^*_v \coloneqq \sup_{\mu \in \P_v(\sX)} \frac{|g(\mu)|}{\mu(v)}. \nonumber
\end{align}
\begin{itemize}
\item [(h)] We assume that $\omega_p(r) \rightarrow 0$ and $\omega_c(r) \rightarrow 0$ as $r\rightarrow0$. Moreover, for any $\mu \in \P_v(\sX)$, the functions
    \begin{align}
    \omega_p(\tilde{\rho}_v(\,\cdot\,,\mu)): \P_v(\sX) \rightarrow \R \nonumber \\
    \intertext{and}
    \omega_c(\tilde{\rho}_v(\,\cdot\,,\mu)): \P_v(\sX) \rightarrow \R \nonumber
    \end{align}
    have finite $v$-norm.
\item [(i)] There exists a non-negative real number $B$ such that
\begin{align}
\sup_{(a,\mu) \in \sA \times \P_v(\sX)} \int_{\sX} v^2(y) p(dy|x,a,\mu) \leq B v^2(x). \nonumber
\end{align}
\end{itemize}

\begin{remark}
Suppose that $v = w$. Define a metric $\lambda_{\sX}$ on $\sX$ as follows:
\begin{align}
\lambda_{\sX}(x,y) \coloneqq \begin{cases} 0, & \text{if $x=y$} \\
v(x) + v(y), & \text{if $x \neq y$} \end{cases} \nonumber
\end{align}
For any real-valued measurable function $g$ on $\sX$, define the Lipschitz seminorm of $g$ as:
\begin{align}
\|g\|_{\lambda} \coloneqq \sup_{x \neq y} \frac{|g(x) - g(y)|}{\lambda_{\sX}(x,y)}. \nonumber
\end{align}
Let $\Lip_{\lambda}(1,\R) \coloneqq \{g: \|g\|_{\lambda} \leq 1\}$. Then, for any $\mu, \nu \in \P_v(\sX)$, we have \cite[Lemma 2.1]{HaMa11}
\begin{align}
\|\mu - \nu\|_v = \sup_{g \in \Lip_{\lambda}(1,\R)} \biggl| \int_{\sX} g(x) \mu(dx) - \int_{\sX} g(x) \nu(dx) \biggr| \label{v-norm}.
\end{align}
This alternative formulation of $v$-norm will be useful when verifying Assumption~2-(h) for specific examples.
\end{remark}

\begin{remark}
In the remainder of this paper, all the proofs are obtained under the assumption that the cost function $c$ is unbounded. The bounded case can be covered by slight modification of the proofs for the unbounded case.
\end{remark}

Before proceeding to the next section, we prove an important (but straightforward) result which will be used in the sequel. It basically states that there is no loss of generality in restricting the infima in Def.~\ref{def1} to weakly continuous Markov policies.

\begin{theorem}\label{theorem0}
Suppose Assumption~1 holds. Then, for any policy ${\boldsymbol \pi}^{(N)} \in {\bf \sM}^{(N)}$, we have
\begin{align}
\inf_{\pi^i \in \sM_i} J_i^{(N)}({\boldsymbol \pi}^{(N)}_{-i},\pi^i) = \inf_{\pi^i \in \sM_i^c} J_i^{(N)}({\boldsymbol \pi}^{(N)}_{-i},\pi^i) \nonumber
\end{align}
for each $i=1,\ldots,N$.
\end{theorem}

\begin{proof}
The proof is given in Appendix~\ref{app0}.
\end{proof}

\section{Mean-field games and mean-field equilibria}\label{sec3}

We begin by considering a mean-field game that can be interpreted as the infinite-population limit $N \to \infty$ of the game introduced in the preceding section. This mean-field game is specified by the quintuple $\bigl( \sX, \sA, p, c, \mu_0 \bigr)$, where, as before, $\sX$ and $\sA$ denote the state and action spaces, respectively, $p(\,\cdot\,|x,a,\mu)$ is the transition probability, and $c$ is the one-stage cost function. We also define the history spaces as $\sG_0 = \sX$ and $\sG_{t}=(\sX\times\sA)^{t}\times \sX$ for $t=1,2,\ldots$, which are endowed with their product Borel $\sigma$-algebras. A \emph{policy} is a sequence $\pi=\{\pi_{t}\}$ of stochastic kernels on $\sA$ given $\sG_{t}$. The set of all policies is denoted by $\Pi$. A \emph{Markov} policy is a sequence $\pi=\{\pi_{t}\}$ of stochastic kernels on $\sA$ given $\sX$. The set of Markov policies is denoted by $\sM$.

\begin{remark}
It is important to note that mean-field games are not games in the strict sense. As will be shown below, they are single-agent stochastic control problems with a constraint on the distribution of the state at each time step.
\end{remark}

In this section, we impose Assumption~1 on the components $\bigl( \sX, \sA, p, c, \mu_0 \bigr)$ of the mean-field game model.

Instead of $N$ agents in the original game, here we have a single agent and model the collective behavior of (a large population of) other agents by an exogenous \textit{state-measure flow} $\bmu := (\mu_t)_{t \ge 0} \subset \P(\sX)$  with a given initial condition $\mu_0$. We say that a policy $\pi^{*} \in \Pi$ is optimal for $\bmu$ if
\begin{align}
J_{\bmu}(\pi^{*}) = \inf_{\pi \in \Pi} J_{\bmu}(\pi), \nonumber
\end{align}
where
\begin{align}
J_{\bmu}(\pi) &\coloneqq  E^{\pi}\biggl[ \sum_{t=0}^{\infty} \beta^t c(x(t),a(t),\mu_t) \biggr] \nonumber
\end{align}
is the infinite-horizon discounted cost of policy $\pi$ with the measure flow $\bmu$. Here, the evolution of the states and actions is given by
\begin{align}
x(0) &\sim \mu_0, \nonumber \\
x(t) &\sim p(\,\cdot\,|x(t-1),a(t-1),\mu_{t-1}), \text{ } t=1,2,\ldots \nonumber \\
a(t) &\sim \pi_t(\,\cdot\,|g(t)), \text{ } t=0,1,\ldots, \nonumber
\end{align}
where $g(t) \in \sG_t$ is the state-action history  up to time $t$.

Let $\M \coloneqq \bigl\{\bmu \in \P(\sX)^{\infty}: \mu_0 \text{ is fixed}\bigr\}$ be the set of all state-measure flows with a given initial condition $\mu_0$. Define the set-valued mapping $\Phi : \M \rightarrow 2^{\Pi}$ 
as $\Phi({\boldsymbol \mu}) = \{\pi \in \Pi: \pi \text{ is optimal for } {\boldsymbol \mu}\}$. Conversely, we define a mapping $\Lambda : \Pi \to \M$ as follows: given $\pi \in \Pi$, the state-measure flow $\bmu := \Lambda(\pi)$ is constructed recursively as
\begin{align}
\mu_{t+1}(\,\cdot\,) = \int_{\sX \times \sA} p(\,\cdot\,|x(t),a(t),\mu_t) \rP^{\pi}(da(t)|x(t)) \mu_t(dx(t)), \nonumber
\end{align}
where $\rP^{\pi}(da(t)|x(t))$ denotes the conditional distribution of $a(t)$ given $x(t)$ under $\pi$ and $(\mu_{\tau})_{0\leq\tau\leq t}$. Note that if $\pi$ is a Markov policy (i.e., $\pi_t(da(t)|g(t)) = \pi_t(da(t)|x(t))$ for all $t$), then $\rP^{\pi}(da(t)|x(t)) = \pi_t(da(t)|x(t))$.

We are now in a position to introduce the notion of an equilibrium for the mean-field game:
\begin{definition}
A pair $(\pi,{\boldsymbol \mu}) \in \Pi \times \M$ is a \emph{mean-field equilibrium} if $\pi \in \Phi({\boldsymbol \mu})$ and $\bmu = \Lambda(\pi)$.
\end{definition}

The following structural result shows that the restriction to Markov policies entails no loss of optimality:
\begin{proposition}\label{lemma1}
For any state measure flow $\bmu \in \M$, we have
\begin{align}
\inf_{\pi \in \Pi} J_{\bmu}(\pi) = \inf_{\pi \in \sM} J_{\bmu}(\pi). \nonumber
\end{align}
Furthermore, we have $\Lambda(\Pi) = \Lambda(\sM)$; that is, for any $\pi \in \Pi$, there exists $\hat{\pi} \in \sM$ such that $\mu_t^{\pi} = \mu_t^{\hat{\pi}}$ for all $t\geq0$.
\end{proposition}

\begin{proof}
The proof is given in Appendix~\ref{app01}.
\end{proof}

\noindent In other words, we can restrict ourselves to Markov policies in the definitions of $\Phi$ and $\Lambda$ without loss of generality --- that is, we have $\Phi(\M) = 2^{\sM}$ and $\Lambda(\sM) = \M$. This implies that, unlike the finite-population case, the policy in the mean-field equilibrium, if it exists, constitutes a true Nash equilibrium of Markovian type for the infinite-population game problem. Put differently, this policy is player-by-player optimal in the class of \emph{all admissible} (not only Markov) policies.

\noindent The main result of this section is the existence of a mean-field equilibrium under Assumption~1.

\begin{theorem}\label{thm:MFE} Under Assumption~1, the mean-field game $(\sX,\sA,p,c,\mu_0)$ admits a mean-field equilibrium $(\pi,\bnu) \in \sM \times \M$.
\end{theorem}

\subsection*{Proof of Theorem~\ref{thm:MFE}}\label{sec3-1}

For each $t\geq0$, let us define
\begin{align}
L_t \coloneqq \sum_{k=t}^{\infty} (\beta \alpha)^{k-t} M_k, \nonumber
\end{align}
where $\alpha$ and $M_k$ are constants defined in Assumption~1-(c) and (f), respectively. Note that
\begin{align}
L_t = M_t + (\beta \alpha) L_{t+1}. \nonumber
\end{align}
For each $t\geq0$, we define
\begin{align}
C_v^t(\sX) &\coloneqq \bigl\{ u \in C_v(\sX): \|u\|_v \leq L_t \bigr\} \nonumber \\
\intertext{and}
\P_v^t(\sX \times \sA) &\coloneqq \bigl\{ \mu \in \P(\sX \times \sA): \mu_1 \in \P_v^t(\sX) \bigr\}, \nonumber
\end{align}
where for any $\nu \in \P(\sX \times \sA)$, $\nu_1$ denotes the marginal of $\nu$ on $\sX$; that is,
\begin{align}
\nu_1(\,\cdot\,) \coloneqq \nu(\,\cdot\, \times \sA). \nonumber
\end{align}
Moreover, we define
\begin{align}
\C &\coloneqq \prod_{t=0}^{\infty} C_v^t(\sX) \nonumber \\
\Xi &\coloneqq \prod_{t=0}^{\infty} \P_v^t(\sX \times \sA). \nonumber
\end{align}
We equip $\C$ with the following metric:
\begin{align}
\rho(\bu,\bv) \coloneqq \sum_{t=0}^{\infty} \sigma^{-t} \|u_t - v_t\|_v, \nonumber
\end{align}
where $\sigma>0$ is chosen so that $\sigma > \gamma$ and $\sigma \beta \alpha <1$. The first condition and Assumption~1-(g) guarantee that $\rho(\bu,\bv) <\infty$ for all $\bu,\bv \in \C$. It can also be proved that $\C$ is complete with respect to $\rho$.

For any $\bnu \in \Xi$ and $t\geq0$, we define the operator $T_t^{\bnu}$ as
\begin{align}
T_t^{\bnu} u(x) = \min_{a \in \sA} \biggl[ c(x,a,\nu_{t,1}) + \beta \int_{\sX} u(y) p(dy|x,a,\nu_{t,1}) \biggr], \nonumber
\end{align}
where $u: \sX \rightarrow \R$.

\begin{lemma}\label{newlemma1}
Let $\bnu \in \Xi$ be arbitrary. Then, for all $t\geq0$, $T_t^{\bnu}$ maps $C_v^{t+1}(\sX)$ into $C_v^t(\sX)$. In addition, for any $u$, $r \in C_v^{t+1}(\sX)$, we have
\begin{align}
\|T_t^{\bnu} u - T_t^{\bnu} r\|_v \leq \alpha \beta \|u-r\|_v. \label{neweq1}
\end{align}
\end{lemma}

\begin{proof}
Let $u \in C_v^{t+1}(\sX)$. By \cite[Proposition 7.32]{BeSh78}, $T_t^{\bnu}u$ is continuous. We also have
\begin{align}
\|T_t^{\bnu}u\|_v
&\leq \sup_{(x,a) \in \sX \times \sA} \frac{\biggl| c(x,a,\nu_{t,1}) + \beta \int_{\sX} u(y) p(dy|x,a,\nu_{t,1}) \biggr|}{v(x)} \nonumber\\
&\leq \sup_{(x,a) \in \sX \times \sA} \frac{M_t v(x) + \beta \alpha L_{t+1} v(x)}{v(x)} \nonumber \\
&=  M_t + \beta \alpha L_{t+1} = L_t, \nonumber
\end{align}
where the last inequality follows from Assumption~1-(c) and (f). This completes the proof of the first statement. The proof of the second statement is straightforward, so we omit the details.
\end{proof}

Using operators $\{T_t^{\bnu}\}_{t\geq0}$, let us define the operator $T^{\bnu}: \C \rightarrow \C$ as
\begin{align}
\bigl( T^{\bnu} \bu \bigr)_t = T_t^{\bnu} u_{t+1}, \text{ for } t\geq0. \label{eq4}
\end{align}
By Lemma~\ref{newlemma1}, $T^{\bnu}$ is a well defined operator; that is, it maps $\C$ into itself.

Since $T_t^{\bnu}$ satisfies (\ref{neweq1}) for all $t\geq0$, it can be shown that $T^{\bnu}$ is a contraction operator on $\C$ with modulus $\sigma \beta \alpha < 1$. Hence, $T^{\bnu}$ has a unique fixed point by the Banach fixed point theorem, as $\C$ is complete with respect to $\rho$.

For any $\bnu \in \Xi$, we let $J_{*,t}^{\bnu}: \sX \rightarrow \R$ denote the discounted-cost value function at time $t$ of the nonhomogeneous Markov decision process with the one-stage cost functions $\bigl\{c(x,a,\nu_{t,1})\bigr\}_{t\geq0}$ and the transition probabilities $\bigl\{p(\,\cdot\,|x,a,\nu_{t,1})\bigr\}_{t\geq0}$. Under Assumption~1, it can be proved that, for all $t\geq0$, $J_{*,t}^{\bnu}$ is continuous. Let $J_{*}^{\bnu} \coloneqq \bigl( J^{\bnu}_{*,t}\bigr)_{t\geq0}$.

\begin{lemma}\label{newlemma2}
For any $\bnu \in \Xi$, we have $J_{*}^{\bnu} \in \C$.
\end{lemma}

\begin{proof}
Let $\pi$ be arbitrary Markov policy. Note that for all $t\geq0$, we have
\begin{align}
J^{\bnu}_{*,t}(y) &\leq \sum_{k=t}^{\infty}  \beta^{k-t} E^{\pi}\bigl[ c(x(k),a(k),\nu_{k,1}) \bigr| x(t)=y \bigr] \nonumber \\
&\leq \sum_{k=t}^{\infty} \beta^{k-t} M_k E^{\pi}\bigl[ v(x(k)) \bigr| x(t)=y \bigr] \text{ }(\text{by Assumption~1-(f)}) \nonumber \\
&\leq \sum_{k=t}^{\infty} \beta^{k-t} M_k \alpha^{k-t} v(y) \text{ }(\text{by Assumption~1-(c)}) \nonumber \\
&= L_t v(y). \nonumber
\end{align}
Hence, $J^{\bnu}_{*,t} \in C_v^t(\sX)$. This completes the proof.
\end{proof}

The following theorem is a known result in the theory of nonhomogeneous Markov decision processes (see \cite[Theorems 14.4 and 17.1]{Hin70}).

\begin{theorem}\label{theorem1}
For any $\bnu \in \Xi$, the collection of value functions $J^{\bnu}_{*}$ is the unique fixed point of the operator $T^{\bnu}$. Furthermore, $\pi \in \sM$ is optimal if and only if
\begin{align}
\nu_t^{\pi} \biggl( \biggr\{ (x,a) : c(x,a,\nu_{t,1}) + &\beta \int_{\sX} J_{*,t+1}^{\bnu}(y) p(dy|x,a,\nu_{t,1}) \nonumber \\
&= T_t^{\bnu} J_{*,t+1}^{\bnu}(x) \biggr\} \biggr) = 1, \label{eq5}
\end{align}
where $\nu_t^{\pi} = {\cal L}\bigl( x(t),a(t) \bigr)$ under $\pi$ and $\bnu$.
\end{theorem}

\noindent To prove the existence of a mean-field equilibrium, we adopt the technique of Jovanovic and Rosenthal \cite{JoRo88}. Define the set-valued mapping $\Gamma: \Xi \rightarrow 2^{\P(\sX \times \sA)^{\infty}}$ as follows:
\begin{align}
\Gamma(\bnu) = C(\bnu) \cap B(\bnu), \nonumber
\end{align}
where
\begin{align}
C(\bnu) &\coloneqq \biggl\{ \bnu' \in \P(\sX \times \sA)^{\infty}: \nu'_{0,1} = \mu_0 \text{ and } \nonumber \\
&\phantom{xxxxxxxxxxxxxx}\nu'_{t+1,1}(\,\cdot\,) = \int_{\sX \times \sA} p(\,\cdot\,|x,a,\nu_{t,1}) \nu_t(dx,da) \biggr\} \nonumber \\
\intertext{and}
B(\bnu) &\coloneqq \biggl\{ \bnu' \in \P(\sX \times \sA)^{\infty}: \forall t\geq0, \text{ } \nu_t' \biggl( \biggr\{ (x,a) : c(x,a,\nu_{t,1}) \nonumber \\
&\phantom{xxxxxxxx}+ \beta \int_{\sX} J_{*,t+1}^{\bnu}(y) p(dy|x,a,\nu_{t,1}) = T_t^{\bnu} J^{\bnu}_{*,t+1}(x) \biggr\} \biggr) = 1 \biggr\}. \nonumber
\end{align}
The following proposition implies that the image of $\Xi$ under $\Gamma$ is contained in $2^{\Xi}$.

\begin{proposition}\label{prop2}
For any $\bnu \in \Xi$, we have $\Gamma(\bnu) \subset \Xi$.
\end{proposition}

\begin{proof}
Fix any $\bnu \in \Xi$. It is sufficient to prove that $C(\bnu) \subset \Xi$. Let $\bnu' \in C(\bnu)$. We prove by induction that $\nu'_{t,1} \in \P^t_v(\sX)$ for all $t\geq0$. The claim trivially holds for $t=0$ as $\nu'_{0,1} = \mu_0$. Assume the claim holds for $t$ and consider $t+1$. We have
\begin{align}
\int_{\sX} w(y) \nu'_{t+1,1}(dy) &= \int_{\sX \times \sA} \int_{\sX} w(y) p(dy|x,a,\nu_{t,1}) \nu_{t}(dx,da) \nonumber \\
&\leq \int_{\sX} \alpha w(x) \nu_{t,1}(dx) \nonumber \text{ }(\text{by Assumption~1-(c)}) \\
&\leq \alpha^{t+1} M \text{ }(\text{as $\nu_{t,1} \in \P^t_v(\sX)$}). \nonumber
\end{align}
Hence, $\nu'_{t+1,1} \in \P^{t+1}_v(\sX)$. This completes the proof.
\end{proof}

We say that $\nu \in \Xi$ is a fixed point of $\Gamma$ if $\bnu \in \Gamma(\bnu)$. The following proposition makes the connection between mean-field equilibria and the fixed points of $\Gamma$.

\begin{proposition}\label{prop1}
Suppose that $\Gamma$ has a fixed point $\bnu = (\nu_t)_{t \ge 0}$. Construct a Markov policy $\pi = (\pi_t)_{t \ge 0}$ by disintegrating each $\nu_t$ as $\nu_t(dx,da) = \nu_{t,1}(dx) \pi_t(da|x)$, and let $\bnu_1 = (\nu_{t,1})_{t \ge 0}$. Then the pair $(\pi,\bnu_1)$ is a mean-field equilibrium.
\end{proposition}

\begin{proof}
If $\bnu \in \Gamma(\bnu)$, then corresponding Markov policy $\pi$ satisfies (\ref{eq5}) for $\bnu$. Therefore, by Theorem~\ref{theorem1}, $\pi \in \Phi(\bnu_1)$. Moreover, since $\bnu \in C(\bnu)$, we have $\Lambda(\pi) = \nu_1$. This completes the proof.
\end{proof}

\noindent By Proposition~\ref{prop1}, it suffices to prove that $\Gamma$ has a fixed point in order to establish the existence of a mean-field equilibrium. To that end, we will use Kakutani's fixed point theorem \cite[Corollary 17.55]{AlBo06}. Note that, since $w$ is a continuous moment function, the set $\P^t_v(\sX)$ is compact with respect to the weak topology \cite[Proposition E.8, p. 187]{HeLa96}, and so, $\P^t_v(\sX \times \sA)$ is tight as $\sA$ is compact. Furthermore, $\P^t_v(\sX \times \sA)$ is closed with respect to the weak topology. Hence, $\P^t_v(\sX \times \sA)$ is compact with respect to the weak topology. Therefore, $\Xi$ is compact with respect to the product topology. In addition, $\Xi$ is also convex.

Note that it can be proved in the same way as in \cite[Theorem 1]{JoRo88} that $C(\bnu) \cap B(\bnu) \neq \emptyset$ for any $\bnu \in \Xi$. Furthermore, we can show that both $C(\bnu)$ and $B(\bnu)$ are convex, and thus their intersection is also convex. Moreover, $\Xi$ is a convex compact subset of a locally convex topological space $\M(\sX \times \sA)^{\infty}$, where $\M(\sX \times \sA)$ denotes the set of finite signed measures on $\sX \times \sA$. The final piece we need in order to deduce the existence of a fixed point of $\Gamma$ by an appeal to Kakutani's fixed point theorem is the following:

\begin{proposition}\label{prop3}
The graph of $\Gamma$, i.e., the set
	$$
	\Gr(\Gamma) := \left\{ (\bnu,\bxi) \in \Xi \times \Xi : \bxi \in \Gamma(\bnu)\right\},
	$$
is closed.
\end{proposition}

\begin{proof}
Let $\bigl\{(\bnu^{(n)},\bxi^{(n)})\bigr\}_{n\geq1} \subset \Xi \times \Xi$ be such that $\bxi^{(n)} \in \Gamma(\bnu^{(n)})$ for all $n$ and $(\bnu^{(n)},\bxi^{(n)}) \rightarrow (\bnu,\bxi)$ as $n\rightarrow\infty$ for some $(\bnu,\bxi) \in \Xi \times \Xi$. To prove $\Gr(\Gamma)$ is closed, it is sufficient to prove $\bxi \in \Gamma(\bnu)$.

We first prove that $\bxi \in C(\bnu)$. For all $n$ and $t$, we have
\begin{align}
\xi^{(n)}_{t+1,1}(\,\cdot\,) = \int_{\sX \times \sA} p(\,\cdot\,|x,a,\nu^{(n)}_{t,1}) \nu^{(n)}_t(dx,da). \label{eq6}
\end{align}
Since $\bxi^{(n)} \rightarrow \bxi$ in $\Xi$, $\xi^{(n+1)}_{t+1} \rightarrow \xi_{t+1}$ weakly. Let $g \in C_b(\sX)$. Then, by \cite[Theorem 3.3]{Ser82}, we have
\begin{align}
\lim_{n\rightarrow\infty} \int_{\sX \times \sA} \int_{\sX}& g(y) p(dy|x,a,\nu^{(n)}_{t,1}) \nu^{(n)}_t(dx,da) =\int_{\sX \times \sA} \int_{\sX} g(y) p(dy|x,a,\nu_{t,1}) \nu_t(dx,da) \nonumber
\end{align}
since $\bnu^{(n)}_t \rightarrow \bnu_t$ weakly and $\int_{\sX} g(y) p(dy|x,a,\nu^{(n)}_{t,1})$ converges to $\int_{\sX} g(y) p(dy|x,a,\nu_{t,1})$ continuously\footnote{Suppose $g$, $g_n$ ($n\geq1$) are measurable functions on metric space $\sE$. The sequence $g_n$ is said to converge to $g$ continuously if $\lim_{n\rightarrow\infty}g_n(e_n)=g(e)$ for any $e_n\rightarrow e$ where $e \in \sE$.} (see \cite[p. 388]{Ser82}). This implies that the sequence of measures on the right-hand side of (\ref{eq6}) converges weakly to $\int_{\sX \times \sA} p(\,\cdot\,|x,a,\nu_{t,1}) \nu_t(dx,da)$. Therefore, we have
\begin{align}
\xi_{t+1,1}(\,\cdot\,) = \int_{\sX \times \sA} p(\,\cdot\,|x,a,\nu_{t,1}) \nu_t(dx,da), \nonumber
\end{align}
from which we deduce that $\bxi \in C(\bnu)$.

It remains to prove that $\bxi \in B(\bnu)$. For each $n$ and $t$, let us define
\begin{align}
F^{(n)}_t(x,a) &= c(x,a,\nu^{(n)}_{t,1}) + \beta \int_{\sX} J^{\bnu^{(n)}}_{*,t+1}(y) p(dy|x,a,\nu^{(n)}_{t,1}) \nonumber \\
\intertext{and}
F_t(x,a) &= c(x,a,\nu_{t,1}) + \beta \int_{\sX} J^{\bnu}_{*,t+1}(y) p(dy|x,a,\nu_{t,1}). \nonumber
\end{align}
Recall that, by definition,
\begin{align}
J^{\bnu^{(n)}}_{*,t}(x) = \min_{a \in \sA} F^{(n)}_t(x,a) \text{ } \text{ and } \text{ } J^{\bnu}_{*,t}(x) = \min_{a \in \sA} F_t(x,a). \nonumber
\end{align}
By assumption, we have
\begin{align}
1 = \xi^{(n)}_t\biggl( \biggl\{ (x,a): F^{(n)}_t(x,a) = J^{\bnu^{(n)}}_{*,t}(x) \biggr\} \biggr), \text{ } \text{for all $n$}. \nonumber
\end{align}
Let $A_t^{(n)} \coloneqq \bigl\{ (x,a): F^{(n)}_t(x,a) = J^{\bnu^{(n)}}_{*,t}(x) \bigr\}$. Since both $F^{(n)}_t$ and $J^{\bnu^{(n)}}_{*,t}$ are continuous, $A_t^{(n)}$ is closed. Define $A_t \coloneqq \bigl\{ (x,a): F_t(x,a) = J^{\bnu}_{*,t}(x) \bigr\}$ which is also closed as both $F_t$ and $J^{\bnu}_{*,t}$ are continuous.

Suppose that $F_t^{(n)}$ converges to $F_t$ continuously and $J^{\bnu^{(n)}}_{*,t}$ converges to $J^{\bnu}_{*,t}$ continuously, as $n\rightarrow\infty$.

For each $M\geq1$, define $B_t^M \coloneqq \bigl\{ (x,a): F_t(x,a) \geq J^{\bnu}_{*,t}(x) + \epsilon(M) \bigr\}$ which is closed, where $\epsilon(M) \rightarrow 0$ as $M\rightarrow \infty$. Since both $F_t$ and $J^{\bnu}_{*,t}$ is continuous, we can choose $\{\epsilon(M)\}_{M\geq1}$ so that $\xi_t(\partial B_t^M) = 0$ for each $M$. Since $A_t^c = \bigcup_{M=1}^{\infty} B_t^M$ and $B_t^M \subset B_t^{M+1}$, we have by monotone convergence theorem
\begin{align}
\xi^{(n)}_t\big(A_t^c \cap A_t^{(n)}\big) = \liminf_{M\to\infty} \xi^{(n)}_t\big(B^M_t \cap A_t^{(n)}). \nonumber
\end{align}
Hence, we have
\begin{align}
1 &= \limsup_{n\rightarrow\infty} \liminf_{M\rightarrow\infty} \biggl\{ \xi^{(n)}_t\big(A_t \cap A^{(n)}_t\big) + \xi^{(n)}_t\big(B^M_t \cap A_t^{(n)}\big)\biggr\} \nonumber\\
&\leq \liminf_{M\rightarrow\infty} \limsup_{n\rightarrow\infty}  \biggl\{\xi^{(n)}_t\big(A_t \cap A^{(n)}_t\big) + \xi^{(n)}_t\big(B^M_t \cap A_t^{(n)}\big)\biggr\}. \nonumber
\end{align}
\noindent For fixed $M$, let us evaluate the limit of the second term in the last expression as $n\rightarrow\infty$. First, note that $\xi^{(n)}_t$ converges weakly to $\xi_t$ as $n\rightarrow\infty$ when both measures are restricted to $B_t^M$, as $B_t^M$ is closed and $\xi_t(\partial B_t^M)=0$ \cite[Theorem 8.2.3]{Bog07}. Furthermore, $1_{A^{(n)}_t \cap B^M_t}$ converges continuously to $0$: if $(x^{(n)},a^{(n)}) \rightarrow (x,a)$ in $B_t^M$, then
\begin{align}
\lim_{n\rightarrow\infty} F_t^{(n)}(x^{(n)},a^{(n)}) &= F_t(x,a) \nonumber \\
&\geq J_{*,t}(x) + \epsilon(M) \nonumber \\
&= \lim_{n\rightarrow\infty} J^{(n)}_{*,t}(x^{(n)}) + \epsilon(M). \nonumber
\end{align}
Hence, for large enough $n$'s, we have $F_t^{(n)}(x^{(n)},a^{(n)}) > J^{(n)}_{*,t}(x^{(n)})$ which implies that $(x^{(n)},a^{(n)}) \not\in A^{(n)}_t$. Then, by \cite[Theorem 3.3]{Ser82}, for each $M$ we have
\begin{align*}
\limsup_{n\rightarrow\infty} \xi^{(n)}_t\big(B^M_t \cap A^{(n)}_t\big) = 0.
\end{align*}
Therefore, we obtain
\begin{align*}
1 &\leq  \limsup_{n\rightarrow\infty} \xi^{(n)}_t\big(A_t \cap A_t^{(n)}\big)\\
&\le \limsup_{n\rightarrow\infty} \xi_t^{(n)}(A_t) \nonumber \\
&\leq \xi_t(A_t), \nonumber
\end{align*}
where the last inequality follows from Portmanteau theorem \cite[Theorem 2.1]{Bil99} and the fact that $A_t$ is closed. Hence, $\xi_t(A_t)=1$. Since $t$ is arbitrary, this is true for all $t$. This means that $\bxi \in B(\bnu)$. Therefore, $\bxi \in \Gamma(\bnu)$ which completes the proof under the assumption that $F_t^{(n)}$ converges to $F_t$ continuously and $J^{\bnu^{(n)}}_{*,t}$ converges to $J^{\bnu}_{*,t}$ continuously, as $n\rightarrow\infty$, which we prove next.
\end{proof}

Note that, for continuous functions, continuous convergence coincides with the uniform convergence over compact sets (see \cite[Lemma 2.1]{Lan81}). Therefore, it is equivalent to establish that $F^{(n)}_t$ uniformly converges to $F_t$ over compact sets and $J^{\bnu^{(n)}}_{*,t}$ uniformly converges to $J^{\bnu}_{*,t}$ over compact sets, as these functions are all continuous.
Furthermore, if $J^{\bnu^{(n)}}_{*,t}$ converges to $J^{\bnu}_{*,t}$ continuously for all $t$, then $F^{(n)}_t$ also converges to $F_t$ continuously for all $t$. Indeed, let $(x^{(n)},a^{(n)}) \rightarrow (x,a)$. Since $J^{\bnu^{(n)}}_{*,t+1}(y) \leq L_{t+1} v(y)$ for all $n\geq1$ and $\int_{\sX} v(y) p(dy|x^{(n)},a^{(n)},\nu^{(n)}_{t,1}) \rightarrow \int_{\sX} v(y) p(dy|x,a,\nu_{t,1})$, by \cite[Theorem 3.3]{Ser82} we have
\begin{align}
\lim_{n\rightarrow\infty} F_t^{(n)}(x^{(n)},a^{(n)}) &= \lim_{n\rightarrow\infty} \biggl[ c(x^{(n)},a^{(n)},\nu^{(n)}_{t,1}) + \beta \int_{\sX} J^{\bnu^{(n)}}_{*,t+1}(y) p(dy|x^{(n)},a^{(n)},\nu^{(n)}_{t,1}) \biggr] \nonumber \\
&= c(x,a,\nu_{t,1}) + \beta \int_{\sX} J^{\bnu}_{*,t+1}(y) p(dy|x,a,\nu_{t,1}) \nonumber \\
&= F_t(x,a). \nonumber
\end{align}
Therefore, it is sufficient to prove that $J^{\bnu^{(n)}}_{*,t}$ uniformly converges to $J^{\bnu}_{*,t}$ over compact sets, for all $t$.

\begin{proposition}\label{prop4}
For any compact $K\subset \sX$, we have
\begin{align}
\lim_{n\rightarrow\infty} \sup_{x \in K} \bigl| J^{(n)}_{*,t}(x) - J_{*,t}(x)| = 0 \nonumber
\end{align}
for all $t\geq0$. Therefore, $J^{(n)}_{*,t}$ converges to $J_{*,t}$ continuously as $n\rightarrow\infty$, for all $t$.
\end{proposition}

\begin{proof}
The proof of the proposition is given in Appendix~\ref{app1}.
\end{proof}

\noindent Theorem~\ref{thm:MFE} is now a consequence of the following:

\begin{theorem}\label{theorem2}
Under Assumption~1, there exists a fixed point $\bnu$ of the set valued mapping $\Gamma: \Xi \rightarrow 2^{\Xi}$. Therefore, the pair $(\pi,\bnu_1)$ is a mean field equilibrium, where $\pi$ and $\bnu_1$ are constructed as in the statement of Proposition~\ref{prop1}.
\end{theorem}

\begin{proof}
Recall that $\Xi$ is a compact convex subset of the locally convex topological space $\M(\sX \times \sA)^{\infty}$. Furthermore, $\Gamma$ has closed graph by Proposition~\ref{prop3}, and it takes nonempty convex values. Therefore, by Kakutani's fixed point theorem \cite[Corollary 17.55]{AlBo06}, $\Gamma$ has a fixed point. The second statement follows from Proposition~\ref{prop1}.
\end{proof}

\section{Existence of Approximate Markov-Nash Equilibria}\label{sec4}

Now we are in a position to prove the main result of the paper --- namely, the existence of approximate Markov-Nash equilibria in games with sufficiently many agents. Let $(\pi,\bmu)$ denote the mean-field equilibrium, which exists by Theorem~\ref{thm:MFE}. In a nutshell, the proof boils down to showing that, if each of the $N$ agents adopts the mean-field equilibrium policy $\pi$, then the resulting policy $\bpi^{(N)} = \{\pi,\pi,\ldots,\pi\}$ is an $\varepsilon$-Markov-Nash equilibrium for all sufficiently large $N$.

In addition to Assumption~1, we impose Assumption~2 in this section. Furthermore, we assume that
\begin{itemize}
\item [(j)] For each $t\geq0$, $\pi_t: \sX \rightarrow \P(\sA)$ is weakly continuous.
\end{itemize}

\noindent The following theorem is the main result of the paper:

\begin{theorem}\label{theorem4} Suppose that Assumptions~1 and 2, and (j) hold. Then, for
any $\varepsilon>0$, there exists a positive integer $N(\varepsilon)$, such that, for each $N\geq N(\varepsilon)$, the policy ${\boldsymbol \pi}^{(N)} = \{\pi,\pi,\ldots,\pi\}$ is an $\varepsilon$-Markov-Nash equilibrium for the game with $N$ agents.
\end{theorem}

The remainder of the section is devoted to the proof of Theorem~\ref{theorem4}. In a nutshell, the logic of the proof can be described as follows: We first show that, as $N \to \infty$, the empirical distribution of the agents' states at each time $t$ converges to a deterministic limit given by the mean-field equilibrium distribution of the state at time $t$. This allows us to deduce that the evolution of the state of a generic agent closely tracks the equilibrium state-measure flow in the infinite-population limit. We then show that the infinite-population limit is insensitive to individual-agent deviations from the mean-field equilibrium policy.

 We start by defining a sequence of stochastic kernels $\bigl\{P_t^{\pi}(\,\cdot\,|x,\mu)\bigr\}_{t\geq0}$ on $\sX$ given $\sX \times \P(\sX)$ as
\begin{align}
P_t^{\pi}(\,\cdot\,|x,\mu) \coloneqq \int_{\sA} p(\,\cdot\,|x,a,\mu) \pi_t(da|x). \nonumber
\end{align}
Since $\pi_t$ is assumed to be weakly continuous, $P_t^{\pi}(\,\cdot\,|x,\mu)$ is also weakly continuous in $(x,\mu)$. In the sequel, to ease the notation, we will also write $P_t^{\pi}(\,\cdot\,|x,\mu)$ as $P_{t,\mu}^{\pi}(\,\cdot\,|x)$. Recall that measure flow $\bmu$ in the mean-field equilibrium satisfies
\begin{align}
\mu_{t+1}(\,\cdot\,) &= \int_{\sX} P_t^{\pi}(\,\cdot\,|x,\mu_t) \mu_t(dx) = \mu_t P_{t,\mu_t}^{\pi}(\,\cdot\,). \nonumber
\end{align}
For each $N\geq1$, let $\bigl\{x_i^{N}(t)\bigr\}_{1\leq i\leq N}$ denote the state configuration at time $t$ in the $N$-person game under the policy ${\boldsymbol \pi}^{(N)} $, and let $e^{(N)}_t$ denote the corresponding empirical distribution.

\begin{lemma}\label{newlemma3}
Fix an arbitrary $N\geq1$. Let $\bpi^{(N')}$ be arbitrary policy for the $N$-agent game problem. Then, under $\bpi^{(N')}$, for all $t\geq0$ and $i=1,\ldots,N$, we have
\begin{align}
{\cal L}(x_i^N(t)) \in \P_v^t(\sX). \nonumber
\end{align}
\end{lemma}
\begin{proof}
The proof of the lemma is similar to the proof of Proposition~\ref{prop2}.
\end{proof}

Recall the Polish space $(\P_v(\sX),\rho_v)$. Define the Wasserstein distance of order 1 on the set of probability measures $\P(\P_v(\sX))$ over $\P_v(\sX)$ as follows \cite[Definition 6.1]{Vil09}:
\begin{align}
W_1(\Phi,\Psi) \coloneqq \inf \bigl\{ E[\rho_v(X,Y)]: {\cal L}(X) = \Phi \text{ and } {\cal L}(Y) = \Psi \bigr\}. \nonumber
\end{align}
Moreover, define the following spaces:
\begin{align}
\P_1(\P_v(\sX)) &\coloneqq \biggl\{ \Phi \in \P(\P_v(\sX)): \int_{\P_v(\sX)} \rho_v(\mu,\mu_0) \Phi(d\mu) <\infty \biggr\} \nonumber \\
\intertext{and}
C_v(\P_v(\sX)) &\coloneqq \bigl\{ \Upsilon: \P_v(\sX) \rightarrow \R; \Upsilon \text{ is continuous and } \|\Upsilon\|^*_v <\infty \bigr\}. \nonumber
\end{align}
Lemma~\ref{newlemma3} implies that ${\cal L}(e_t^{(N)})(\P_v(\sX)) = 1$ and ${\cal L}(e_t^{(N)})(\,\cdot\,) \in \P_1(\P_v(\sX))$, for all $N\geq1$ and $t\geq0$.

\begin{lemma}\label{lemma5}
Let $\{\Phi_n\}_{n\geq1} \subset \P_1(\P_v(\sX))$ and $\delta_{\mu} \in \P_1(\P_v(\sX))$. Then the following are equivalent.
\begin{itemize}
\item[(i)] $W_1(\Phi_n,\delta_{\mu}) \rightarrow 0$ as $n\rightarrow\infty$.
\item[(ii)] $E\bigl[|F(X_n) - F(X)|\bigr] \rightarrow 0$ as $n\rightarrow\infty$, for any $F \in C_v(\P_v(\sX))$ and for any sequence of $\P_v(\sX)$-valued random elements $\{X_n\}_{n\geq1}$ and a $\P_v(\sX)$-valued random element $X$ such that ${\cal L}(X_n) = \Phi_n$ and ${\cal L}(X) = \delta_{\mu}$.
\item[(iii)] $E\bigl[|X_n(f) - X(f)|\bigr] \rightarrow 0$ as $n\rightarrow\infty$, for any $f \in C_b(\sX) \cup \{v\}$ and for any sequence of $\P_v(\sX)$-valued random elements $\{X_n\}_{n\geq1}$ and a $\P_v(\sX)$-valued random element $X$ such that ${\cal L}(X_n) = \Phi_n$ and ${\cal L}(X) = \delta_{\mu}$.
\end{itemize}
\end{lemma}

\begin{proof}
$(i) \Rightarrow (iii)$ \\
Fix any $\{X_n\}_{n\geq1}$ and $X$ that satisfy the hypothesis of (iii). We first prove the result for $f \in C_b(\sX)$. Define $\Delta_n(\,\cdot\,) \coloneqq {\cal L}(X_n,X)(\,\cdot\,)$. Since ${\cal L}(X) = \delta_{\mu}$, the only coupling between $\Phi_n$ and $\delta_{\mu}$ is given by $\Phi_n \otimes \delta_{\mu}$. Therefore, $\Delta_n = \Phi_n \otimes \delta_{\mu}$. Then, we have
\begin{align}
\lim_{N\rightarrow\infty} E\bigl[ |X_n(f) - X(f)| \bigr] &=\lim_{N\rightarrow\infty} \int_{\P_v(\sX)^2} |\nu(f) - \zeta(f)| \Delta_n(d\nu,d\zeta) \nonumber \\
&=\lim_{N\rightarrow\infty} \int_{\P_v(\sX)} |\nu(f) - \mu(f)| \Phi_n(d\nu) \nonumber \\
&= \int_{\P_v(\sX)} |\nu(f) - \mu(f)| \delta_{\mu}(d\nu) \label{auxxx} \\
&=0, \nonumber
\end{align}
where (\ref{auxxx}) follows from the fact that $h(\nu) \coloneqq |\nu(f) - \mu(f)| \in C_b(\P_v(\sX))$ and $\Phi_n$ converges to $\delta_{\mu}$ weakly. This establishes the result for $f \in C_b(\sX)$. For $f=v$, the result follows from the definition of $\rho_v$ and the fact that $W_1(\Phi_n,\Phi) = E\bigl[ \rho_v(X_n,X) \bigr]$ since there is only one coupling between $\Phi_n$ and $\Phi$.

Note that $(iii), (ii) \Rightarrow (i)$ is clear by definition of $\rho_v$.

$(i) \Rightarrow (ii)$ \\
Let $F \in C_v(\P_v(\sX))$. Since the coupling between $\Phi_n$ and $\delta_{\mu}$ is unique, we can write
\begin{align}
E\bigl[ |F(X_n) - F(X)| \bigr] = \int_{\P_v(\sX)} |F(\nu) - F(\mu)| \Phi_n(d\nu). \nonumber
\end{align}
Define $l(\nu) \coloneqq |F(\nu)-F(\mu)|$. Since $\Phi_n \rightarrow \delta_{\mu}$ weakly, we have \cite[Proposition E.2]{HeLa96}
\begin{align}
\liminf_{n\rightarrow\infty} \int_{\P_v(\sX)} l(\nu) + \|l\|^*_v \nu(v) \Phi_n(d\nu) \geq l(\mu) + \|l\|^*_v \mu(v) \nonumber \\
\intertext{and}
\liminf_{n\rightarrow\infty} \int_{\P_v(\sX)} -l(\nu) + \|l\|^*_v \nu(v) \Phi_n(d\nu) \geq -l(\mu) + \|l\|^*_v \mu(v) \nonumber
\end{align}
as both $l(\nu) + \|l\|^*_v \nu(v)$ and $-l(\nu) + \|l\|^*_v \nu(v)$ are nonnegative continuous functions on $\P_v(\sX)$.
Note that
\begin{align}
\lim_{n\rightarrow\infty} \biggl| \int_{\P_v(\sX)} \nu(v) \Phi_n(d\nu) - \mu(v) \biggr| =0 \nonumber
\end{align}
by definition of $\rho_v$ and (i). Thus, we have
\begin{align}
\liminf_{n\rightarrow\infty} \int_{\P_v(\sX)} l(\nu) \Phi_n(d\nu) &\geq l(\mu) \nonumber \\
&\geq \limsup_{n\rightarrow\infty} \int_{\P_v(\sX)} l(\nu) \Phi_n(d\nu). \nonumber
\end{align}
Therefore,
\begin{align}
\lim_{n\rightarrow\infty} \int_{\P_v(\sX)} l(\nu) \Phi_n(d\nu) = l(\mu) = |F(\mu) - F(\mu)| = 0. \nonumber
\end{align}
This completes the proof.
\end{proof}

The following proposition states that, at each time $t$, the sequence of random measures $e^{(N)}_t$ converges to the mean-field equilibrium distribution $\mu_t$ of the state at time $t$ as $N \to \infty$:

\begin{proposition}\label{prop5}
For all $t\geq0$, $\lim_{N\rightarrow\infty} W_1\bigl({\cal L}(e_t^{(N)}),\delta_{\mu_t}\bigr)=0$ in $\P_1\bigl(\P_v(\sX)\bigr)$.
\end{proposition}

\begin{proof}
We prove that
\begin{align}
\lim_{N\rightarrow\infty} E\bigl[|e_t^{(N)}(f) - \mu_t(f)|\bigr] = 0 \nonumber
\end{align}
for any $f \in C_v(\sX)$ by induction on $t$. Since $C_v(\sX) \supset C_b(\sX) \cup \{v\}$, this will complete the proof by Lemma~\ref{lemma5}.

Note that since $\{x_i^N(0)\}_{1\leq i\leq N} \sim \prod_{i=1}^N \mu_0$, the claim is true for $t=0$ as any $f \in C_v(\sX)$ is $\mu_0$-integrable by Assumption~1-(e). Suppose the claim holds for $t$ and consider $t+1$. Fix any $g \in C_v(\sX)$. Then, we have
\begin{align}
|e_{t+1}^{(N)}(g) - \mu_{t+1}(g)| \leq |e_{t+1}^{(N)}(g) - e_{t}^{(N)} P^{\pi}_{t,e_t^{(N)}}(g)| + |e_t^{(N)} P^{\pi}_{t,e_t^{(N)}}(g) - \mu_t P^{\pi}_{t,\mu_t}(g) |, \label{eq8}
\end{align}
where $\mu_{t+1} = \mu_t P^{\pi}_{t,\mu_t}$ since $(\mu_t)_{t\geq0}$ is the measure flow in the mean field equilibrium.

First, let us consider the second term in (\ref{eq8}). Define $F: \P_v(\sX) \rightarrow \R$ as
\begin{align}
F(\mu) = \mu P^{\pi}_{t,\mu}(g) \coloneqq \int_{\sX} \int_{\sX} g(y) P^{\pi}_t(dy|x,\mu) \mu(dx). \nonumber
\end{align}
Note that
\begin{align}
\|F\|^*_v &= \sup_{\mu \in \P_v(\sX)} \frac{|F(\mu)|}{\mu(v)} \nonumber \\
&= \sup_{\mu \in \P_v(\sX)} \frac{\bigl|\int_{\sX} \int_{\sX} g(y) P^{\pi}_t(dy|x,\mu) \mu(dx)\bigr|}{\mu(v)} \nonumber \\
&= \sup_{\mu \in \P_v(\sX)} \frac{\int_{\sX} \int_{\sX} |g(y)| P^{\pi}_t(dy|x,\mu) \mu(dx)}{\mu(v)} \nonumber \\
&= \sup_{\mu \in \P_v(\sX)} \frac{\int_{\sX} \int_{\sX} \|g\|_v v(y) P^{\pi}_t(dy|x,\mu) \mu(dx)}{\mu(v)} \nonumber \\
&= \|g\|_v \sup_{\mu \in \P_v(\sX)} \frac{\alpha \mu(v)}{\mu(v)} = \alpha \|g\|_v. \nonumber
\end{align}
Hence, $\|F\|^*_v < \infty$. If we can prove that $F$ is also continuous, then by Lemma~\ref{lemma5} the expectation of the second term in (\ref{eq8}) goes to zero. To this end, let $\mu_n \rightarrow \mu$ in $\P_v(\sX)$ with respect to $v$-topology. Define
\begin{align}
l_n(x) &\coloneqq \int_{\sX} g(y) P_t^{\pi}(dy|x,\mu_n) \nonumber \\
\intertext{and}
l(x) &\coloneqq \int_{\sX} g(y) P_t^{\pi}(dy|x,\mu). \nonumber
\end{align}
Note that $F(\mu_n) = \int_{\sX} l_n(x) \mu_n(dx)$ and $F(\mu) = \int_{\sX} l(x) \mu(dx)$. We first prove that $l_n$ converges continuously to $l$. Let $x_n \rightarrow x$. Since $\pi_t$ is assumed to be continuous, we have $\pi_t(\,\cdot\,|x_n) \rightarrow \pi_t(\,\cdot\,|x)$ weakly. Note that
\begin{align}
\sup_{n\geq1} \biggl| \int_{\sX} g(y) p(dy|x_n,a,\mu_n) \biggr| \leq \|g\|_v \alpha \sup_{n\geq1} v(x_n)
< \infty \nonumber
\end{align}
since $v$ is continuous and $\{x_n\}$ is convergent. Therefore, by \cite[Theorem 3.3]{Ser82} we have
\begin{align}
\lim_{n\rightarrow\infty} l_n(x_n) &= \lim_{n\rightarrow\infty} \int_{\sA} \int_{\sX} g(y) p(dy|x_n,a,\mu_n) \pi_t(da|x_n) \nonumber \\
&= \int_{\sA} \int_{\sX} g(y) p(dy|x,a,\mu) \pi_t(da|x) \nonumber \\
&= l(x).
\end{align}
Thus, $l_n$ converges to $l$ continuously. In addition, we have $|l_n(x)| \leq \alpha \|g\|_v v(x)$ for all $x \in \sX$ and $\mu_n(v) \rightarrow \mu(v)$. Then, again by \cite[Theorem 3.3]{Ser82}, we have
\begin{align}
\lim_{n\rightarrow\infty} F(\mu_n) &= \lim_{n\rightarrow\infty} \int_{\sX} l_n(x) \mu_n(dx) \nonumber \\
&= \int_{\sX} l(x) \mu(dx) \nonumber \\
&= F(\mu). \nonumber
\end{align}
Hence, $F \in C_v(\P_v(\sX))$ and so the expectation of the second term in (\ref{eq8}) goes to zero.

Now, consider the first term in (\ref{eq8}). Let us write the expectation of this term as
\begin{align}
E\biggl[ E\biggl[ |e_{t+1}^{(N)}(g) - e_t^{(N)} P^{\pi}_{t,e_t^{(N)}}(g)| \biggr| x_1^N(t),\ldots,x_N^N(t) \biggr] \biggr]. \nonumber
\end{align}
Then, by Lemma~\ref{auxlemma} in Appendix~\ref{newapp1}, we have
\begin{align}
E\biggl[ E\biggl[ |e_{t+1}^{(N)}(g) &- e_t^{(N)} P^{\pi}_{t,e_t^{(N)}}(g)| \biggr| x_1^N(t),\ldots,x_N^N(t) \biggr] \biggr]^2 \nonumber
\\
&\leq E\biggl[ E\biggl[ |e_{t+1}^{(N)}(g) - e_t^{(N)} P^{\pi}_{t,e_t^{(N)}}(g)| \biggr| x_1^N(t),\ldots,x_N^N(t) \biggr]^2 \biggr] \nonumber \\
&\leq E\biggl[ \frac{1}{N^2} \sum_{i=1}^N \biggl\{ \int_{\sX} g^2(y) P^{\pi}_{t,e_t^{(N)}}(dy|x_i^N(t)) \nonumber \\
&\phantom{xxxxxxxxxxxx}+ \biggl( \int_{\sX} g(y) P^{\pi}_{t,e_t^{(N)}}(dy|x_i^N(t)) \biggr)^2 \biggr\}  \biggl] \nonumber \\
&\leq E\biggl[ \frac{\|g\|_v^2}{N^2} \sum_{i=1}^N \biggl\{ \int_{\sX} v^2(y) P^{\pi}_{t,e_t^{(N)}}(dy|x_i^N(t)) \nonumber \\
&\phantom{xxxxxxxxxxxx}+ \biggl( \int_{\sX} v(y) P^{\pi}_{t,e_t^{(N)}}(dy|x_i^N(t)) \biggr)^2 \biggr\}  \biggl] \nonumber \\
&\leq \frac{\|g\|_v^2}{N^2} \sum_{i=1}^N E\bigl[B v^2(x_i^N(t)) + \alpha^2 v^2(x_i^N(t))\bigr]. \nonumber
\end{align}
The last expression follows from Assumption~1-(c) and Assumption~2-(i). For any $t\geq0$, one can prove that $\sup_{N\geq1} \sup_{i=1,\ldots,N} E[v^2(x_i^N(t))] \leq H$ for some $H \in \R$ by Assumption~2-(i). Therefore, the expectation of the first term in (\ref{eq8}) also converges to zero as $N\rightarrow\infty$. Since $g$ is arbitrary, this completes the proof.
\end{proof}

From the above proposition and from the assumed continuity of the transition probability $p(\,\cdot\,|x,a,\mu)$ in $\mu$, we now deduce that the evolution of the state of a generic agent in the original game with sufficiently many agents should closely track the evolution of the state in the mean-field game:

\begin{proposition}\label{prop6} If $(\pi,\bmu)$ is a mean-field equilibrium, then
\begin{align}
\lim_{N\rightarrow\infty} J_1^{(N)}({\boldsymbol \pi}^{(N)}) = J_{\bmu}(\pi) = \inf_{\pi' \in \Pi} J_{\bmu}(\pi'). \nonumber
\end{align}
\end{proposition}

\begin{proof}
For each $t\geq0$, let us define
\begin{align}
c_{\pi_t}(x,\mu) \coloneqq \int_{\sA} c(x,a,\mu) \pi_t(da|x). \nonumber
\end{align}
Note that, for any permutation $\sigma$ which is independent of $x^N_1(t),\ldots,x^N_N(t)$, we have
\begin{align}
{\cal L}\bigl(x_1^N(t),\ldots,x_N^N(t),e_t^{(N)}\bigr) = {\cal L}\bigl(x_{\sigma(1)}^N(t),\ldots,x_{\sigma(N)}^N(t),e_t^{(N)}\bigr). \nonumber
\end{align}
Therefore, we can write
\begin{align}
E\bigl[ c(x_1^N(t),a_1^N(t),e_t^{(N)}) \bigr] &= \frac{1}{N} \sum_{i=1}^N E\bigl[ c(x_i^N(t),a_i^N(t),e_t^{(N)}) \bigr] \nonumber \\
&= E\bigl[ e_t^{(N)}\bigl(c_{\pi_t}(x,e_t^{(N)})\bigr) \bigr]. \nonumber
\end{align}
Define $F: \P_v(\sX) \rightarrow \R$ as
\begin{align}
F(\mu) \coloneqq \int_{\sX} c_{\pi_t}(x,\mu) \mu(dx). \nonumber
\end{align}
Hence, $E\bigl[ c(x_1^N(t),a_1^N(t),e_t^{(N)}) \bigr] = E\bigl[ F(e_t^{(N)}) \bigr]$. First, note that
\begin{align}
\|F&\|^*_v \nonumber \\
&= \sup_{\mu \in \P_v(\sX)} \frac{\bigl| \int_{\sX \times \sA} c(x,a,\mu) \pi_t(da|x)\mu(dx) \bigr|}{\mu(v)} \nonumber \\
&\leq \sup_{\mu \in \P_v(\sX)} \frac{\bigl| \int_{\sX} c_{\pi_t}(x,\mu) \mu(dx) - \int_{\sX} c_{\pi_t}(x,\mu_t) \mu(dx) \bigr|}{\mu(v)} + \sup_{\mu \in \P_v(\sX)} \frac{\bigl| \int_{\sX} c_{\pi_t}(x,\mu_t) \mu(dx) \bigr|}{\mu(v)} \nonumber \\
&\leq \sup_{\mu \in \P_v(\sX)} \frac{\omega_c(\tilde{\rho}_v(\mu,\mu_t))}{\mu(v)} + \sup_{\mu \in \P_v(\sX)} \frac{M_t \mu(v)}{\mu(v)} \nonumber \\
&=\|\omega(\tilde{\rho}_v(\,\cdot\,,\mu_t))\|^*_v + M_t <\infty. \text{ } (\text{by Assumption~2-(h)}) \nonumber
\end{align}
Hence, $F$ has finite $v$-norm. For continuity, let $\mu_n \rightarrow \mu$ in $v$-topology. Then, we have
\begin{align}
|F(\mu_n) - F(\mu)| &\leq \biggl| \int_{\sX} c_{\pi_t}(x,\mu_n) \mu_n(dx) - \int_{\sX} c_{\pi_t}(x,\mu) \mu_n(dx) \biggr| \nonumber \\
&\phantom{xxxxxxxxxxx}+ \biggl| \int_{\sX} c_{\pi_t}(x,\mu) \mu_n(dx) - \int_{\sX} c_{\pi_t}(x,\mu) \mu(dx) \biggr| \nonumber \\
&\leq \omega_c(\tilde{\rho}_v(\mu_n,\mu)) + \biggl| \int_{\sX} c_{\pi_t}(x,\mu) \mu_n(dx) - \int_{\sX} c_{\pi_t}(x,\mu) \mu(dx) \biggr|. \nonumber
\end{align}
Note that $c_{\pi_t}(x,\mu) \in C_v(\sX)$ since $\pi_t$ is weakly continuous and $\mu \in \P_v^l(\sX)$ for some $l\geq1$. Therefore, the last expression goes to zero as $\mu_n \rightarrow \mu$ in $v$-topology. Hence, $F \in C_v(\P_v(\sX))$. Therefore, by Proposition~\ref{prop5} we have
\begin{align}
\lim_{N\rightarrow\infty} E\bigl[ c(x_1^N(t),a_1^N(t),e_t^{(N)}) \bigr] &= \lim_{N\rightarrow\infty} E\bigl[ e_t^{(N)}\bigl(c_{\pi_t}(x,e_t^{(N)})\bigr) \bigr] \nonumber \\
&= \lim_{N\rightarrow\infty} E[F(e_t^{(N)})] \nonumber \\
&= F(\mu_t) \nonumber \\
&= \mu_t(c_{\pi_t}(x,\mu_t)). \label{eq9}
\end{align}
Since $t$ is arbitrary, this is true for all $t\geq0$. Recall that the pair $(\pi,\bmu)$ is a mean field equilibrium, and so we can write \begin{align}
J_{\bmu}(\pi) = \sum_{t=0}^{\infty} \beta^t \mu_t(c_{\pi_t}(x,\mu_t)). \nonumber
\end{align}
Note that for all $N\geq1$
\begin{align}
J_1^{(N)}({\boldsymbol \pi}^{(N)}) &= \sum_{t=0}^{\infty} \beta^t E\bigl[ c(x_1^N(t),a_1^N(t),e_t^{(N)}) \bigr] \nonumber \\
&\phantom{x}\leq \sum_{t=0}^{\infty} \beta^t \biggl\{ E\bigl[ |c(x_1^N(t),a_1^N(t),e_t^{(N)}) - c(x_1^N(t),a_1^N(t),\mu_0)|\bigr] \nonumber \\
&\phantom{xxxxxxxxxxxxxxxxxxxxxxxxxxxx}+ E\bigl[ c(x_1^N(t),a_1^N(t),\mu_0) \bigr] \biggr\} \nonumber \\
&\phantom{x}\leq \sum_{t=0}^{\infty} \beta^t \biggl\{ E\bigl[ \omega_c(\tilde{\rho}_v(e_t^{(N)},\mu_0))\bigr] + E\bigl[ M_0 v(x_1^N(t)) \bigr] \biggr\} \nonumber \\
&\phantom{x}\leq \sum_{t=0}^{\infty} \beta^t \biggl\{ \|\omega_c(\tilde{\rho}_v(\,\cdot\,,\mu_0))\|^*_v E\bigl[ e_t^{(N)}(v)\bigr] + M_0 E\bigl[ v(x_1^N(t)) \bigr] \biggr\} \nonumber \\
&\phantom{x}\leq \bigr (\|\omega_c(\tilde{\rho}_v(\,\cdot\,,\mu_0))\|^*_v + M_0 \bigl) M \sum_{t=0}^{\infty} (\beta \alpha)^t \nonumber \\
&\phantom{x} <\infty. \nonumber
\end{align}
Therefore, by (\ref{eq9}) and the dominated convergence theorem, we obtain
\begin{align}
\lim_{N\rightarrow\infty} J_1^{(N)}({\boldsymbol \pi}^{(N)}) = J_{\bmu}(\pi) \nonumber
\end{align}
which completes the proof.
\end{proof}

Now let $\{\tpi^{(N)}\}_{N\geq1} \subset \sM_1^c$ be an arbitrary sequence of weakly continuous Markov policies for Agent~$1$. For each $N\geq1$, let $\bigl\{\tx_i^N(t)\bigr\}_{1\leq i \leq N}$ be the state configuration at time $t$ in the $N$-person game under the policy $\tilde{{\boldsymbol \pi}}^{(N)} \coloneqq \{\tpi^{(N)},\pi,\ldots,\pi\}$ (i.e., when Agents $2$ through $N$ stick to the mean-field equilibrium policy $\pi$,  while Agent~$1$ deviates with $\tpi^{(N)}$), and let $\te_t^{(N)}$ denote the corresponding empirical distribution. The following result, whose proof is a slight modification of the proof of Proposition~\ref{prop5}, states that, in the infinite-population limit, the law of the empirical distribution of the states at each time $t$ is insensitive to local deviations from the mean-field equilibrium policy:

\begin{proposition}\label{prop8}
For each $t\geq0$, $\lim_{N\rightarrow\infty} W_1\bigl({\cal L}(\te_t^{(N)}),\delta_{\mu_t}\bigr)=0$ in $\P_1(\P_v(\sX))$.
\end{proposition}

\begin{proof}
We follow the same technique that was employed in the proof of Proposition~\ref{prop5}. Namely, we prove that
\begin{align}
\lim_{N\rightarrow\infty} E\bigl[|\te_t^{(N)}(f) - \mu_t(f)|\bigr] = 0 \nonumber
\end{align}
for any $f \in C_v(\sX)$ by induction on $t$.

Since $\bigl\{\tx_i^N(0)\bigr\}_{1\leq i \leq N} \sim \prod_{i=1}^N \mu_0$, the claim holds for $t=0$. Suppose the claim holds for $t$ and consider $t+1$. Choose any $g \in C_v(\sX)$ and write
\begin{align}
|\te_{t+1}^{(N)}(g) - \mu_{t+1}(g)| &\leq |\te_{t+1}^{(N)}(g) - \te_{t+1}^{(N-1)}(g)| \nonumber \\
&+ |\te_{t+1}^{(N-1)}(g) - \te_{t}^{(N-1)}P^{\pi}_{t,\te_t^{(N)}}(g)| \nonumber \\ &+|\te_{t}^{(N-1)}P^{\pi}_{t,\te_t^{(N)}}(g)-\te_{t}^{(N)}P^{\pi}_{t,\te_t^{(N)}}(g)| \nonumber \\
&+ |\te_{t}^{(N)}P^{\pi}_{t,\te_t^{(N)}}(g) - \mu_t P^{\pi}_{t,\mu_t}(g)|, \label{eq10}
\end{align}
where $\te_{t+1}^{(N-1)}$ and $\te_{t}^{(N-1)} $ are given by
\begin{align}
\te_{t+1}^{(N-1)} &\coloneqq \frac{1}{N-1} \sum_{i=2}^N \delta_{\tx_i^{(N)}(t+1)} \nonumber \\
\intertext{and}
\te_{t}^{(N-1)} &\coloneqq \frac{1}{N-1} \sum_{i=2}^N \delta_{\tx_i^{(N)}(t)}. \nonumber
\end{align}
In the remainder of the proof, we show that expectation of each term in (\ref{eq10}) converges to zero as $N\rightarrow\infty$, which will complete the proof as $g$ is arbitrary.

First, let us consider the first term in (\ref{eq10}). We have
\begin{align}
|\te_{t+1}^{(N)}(g) - \te_{t+1}^{(N-1)}(g)| &\leq \biggl| \frac{1}{N} \sum_{i=1}^N g(\tx_i^N(t+1)) -  \frac{1}{N} \sum_{i=2}^N g(\tx_i^N(t+1)) \biggr| \nonumber \\
&\phantom{xxxxxxx}\vee \biggl| \frac{1}{N} \sum_{i=1}^N g(\tx_i^N(t+1)) -  \frac{1}{N-1} \sum_{i=1}^N g(\tx_i^N(t+1)) \biggr| \nonumber \end{align}
When we take the expectation of above term, we obtain
\begin{align}
E\bigl[ |\te_{t+1}^{(N)}(g) &- \te_{t+1}^{(N-1)}(g)| \bigr] \nonumber \\
&\leq \frac{1}{N} E\bigl[ |g(\tx_1^N(t+1))| \bigr] + \biggl| \frac{1}{N} - \frac{1}{N-1} \biggr| \sum_{i=1}^N E\bigl[ |g(\tx_i^N(t+1)) |\bigr] \nonumber \\
&\leq \frac{\|g\|_v}{N} \alpha^{t+1} M + \biggl|1- \frac{N}{N-1}\biggr| \|g\|_v \alpha^{t+1} M. \nonumber
\end{align}
Note that the last expression goes to zero as $N\rightarrow\infty$.

We can write the expectation of the second term in (\ref{eq10}) as
\begin{align}
E\biggl[ E\biggl[ |\te_{t+1}^{(N-1)}(g) - \te_{t}^{(N-1)}P^{\pi}_{t,\te_t^{(N)}}(g)| \biggr| \tx_1^N(t),\ldots,\tx_N^N(t) \biggr] \biggr]. \nonumber 
\end{align}
Using Lemma~\ref{auxlemma} in Appendix~\ref{newapp1}, we can prove that this term converges to zero as in the proof of Proposition~\ref{prop5}.

For the expectation of the third term in (\ref{eq10}), we have
\begin{align}
E\bigl[ |\te_{t}^{(N-1)}P^{\pi}_{t,\te_t^{(N)}}(g)&-\te_{t}^{(N)}P^{\pi}_{t,\te_t^{(N)}}(g)| \bigr] \nonumber \\ &= E\biggl[ \biggl| \int_{\sX} \int_{\sX} g(y) P^{\pi}_{t,\te_t^{(N)}}(dy|x) \te_{t}^{(N-1)}(dx) \nonumber \\
&\phantom{xxxxxxxxxxx}- \int_{\sX} \int_{\sX} g(y) P^{\pi}_{t,\te_t^{(N)}}(dy|x) \te_{t}^{(N)}(dx) \biggr| \biggr] \nonumber \\
&\leq \frac{1}{N} E\biggl[ \biggl| \int_{\sX} g(y) P^{\pi}_{t,\te_t^{(N)}}(dy|\tx_1^N(t)) \biggr| \biggr]   \nonumber \\
&\phantom{xxxxxxxxxxx}+ \biggl| \frac{1}{N} -  \frac{1}{N-1} \biggr| \sum_{i=1}^N E\biggl[ \biggl|\int_{\sX} g(y) P^{\pi}_{t,\te_t^{(N)}}(dy|\tx_i^N(t)) \biggr| \biggr] \nonumber \\
&\leq \frac{\|g\|_v \alpha}{N} E\bigl[ v(\tx_1^N(t)) \bigr] + \biggl|1- \frac{N}{N-1}\biggr| \|g\|_v \alpha \sum_{i=1}^N E\bigl[  v(\tx_i^N(t)) \bigr] \nonumber \\
&\leq \frac{\|g\|_v}{N} \alpha^{t+1} M + \biggl|1- \frac{N}{N-1}\biggr| \|g\|_v \alpha^{t+1} M. \nonumber
\end{align}
Again, the last expression goes to zero as $N\rightarrow\infty$.

For the last term in (\ref{eq10}), let us define $F: \P_v(\sX) \rightarrow \R$ as $F(\mu) \coloneqq \mu P_{t,\mu}^{\pi}(g)$. Since $\pi_t$ is weakly continuous, it can be proven as in the proof of Proposition~\ref{prop5} that $F \in C_V(\P_v(\sX))$. Then, by Lemma~\ref{lemma5}, the expectation of the last term converges to zero since ${\cal L}(\te_t^{(N)}) \rightarrow \delta_{\mu_t}$ with respect to $W_1$ distance in $\P_1\bigl(\P_v(\sX)\bigr)$ by induction hypothesis. This completes the proof.
\end{proof}

\noindent For each $N\geq1$, let $\{\hx^N(t)\}_{t\geq0}$ denote the states of the non-homogenous Markov chain that evolves as follows:
\begin{align}
\hx^N(0) \sim \mu_0 \text{ and } \hx^N(t+1) \sim P^{\tpi^{(N)}}_{t,\mu_t}(\,\cdot\,|\hx^N(t)). \nonumber
\end{align}
Note that
\begin{align}
J_{\bmu}(\tpi^{(N)}) = \sum_{t=0}^{\infty} \beta^t E\bigl[ c(\hx^N(t),\hat{a}^N(t),\mu_t)\bigr]. \nonumber
\end{align}

\begin{lemma}\label{lemma6}
Let ${\cal F}$ be a family of equicontinuous functions on $\P_v(\sX)$ with respect to $v$-topology. For any $t\geq0$, define $F_t: \P_v(\sX) \rightarrow [0,+\infty]$ as:
\begin{align}
F_t(\mu) \coloneqq \sup_{\Psi \in {\cal F}} |\Psi(\mu) - \Psi(\mu_t)|. \nonumber
\end{align}
Suppose that $F_t$ is real valued and $\|F_t\|^*_v < \infty$. Then $F_t \in \C_v(\P_v(\sX))$. Therefore,
\begin{align}
\lim_{N\rightarrow\infty} E\biggl[ \sup_{\Psi \in {\cal F}} \bigl| \Psi(\te_t^{(N)}) - \Psi(\mu_t) \bigr| \biggr] = 0. \nonumber
\end{align}
\end{lemma}

\begin{proof}
It is sufficient to prove that $F_t$ is continuous with respect to $v$-topology. Let $\mu_n \rightarrow \mu$ in $v$-topology. Then we have
\begin{align}
|F_t(\mu_n) - F_t(\mu)| &= \bigl| \sup_{\Psi \in {\cal F}} |\Psi(\mu_n) - \Psi(\mu_t)| - \sup_{\Psi \in {\cal F}} |\Psi(\mu) - \Psi(\mu_t)| \bigr| \nonumber \\
&\leq \sup_{\Psi \in {\cal F}} |\Psi(\mu_n)-\Psi(\mu)| \rightarrow 0 \nonumber
\end{align}
as $N\rightarrow\infty$ since ${\cal F}$ is equicontinuous. This completes the proof.
\end{proof}

Using Lemma~\ref{lemma6} we now prove the following result.

\begin{lemma}\label{lemma7}
Fix any $t\geq0$. Suppose that
\begin{align}
\lim_{N\rightarrow\infty}\bigl| {\cal L}(\tx_1^N(t))(g_N) - {\cal L}(\hx_1^N(t))(g_N) \bigr| = 0 \nonumber
\end{align}
for any sequence $\{g_N\}_{N\geq1} \subset C_v(\sX)$ with $\sup_{N\geq1} \|g_N\|_v < \infty$. Then, we have
\begin{align}
\lim_{N\rightarrow\infty} \bigl| {\cal L}(\tx_1^N(t),\te_t^{(N)})(T_N) - {\cal L}(\hx_1^N(t),\delta_{\mu_t})(T_N) \bigr| = 0 \nonumber
\end{align}
for any sequence $\{T_N\}$ of functions on $\sX \times \P_v(\sX)$ satisfying:
\begin{itemize}
\item [(i)] The family $\bigl\{T_N(x,\,\cdot\,): x \in \sX, N\geq1)\bigr\}$ is equicontinuous with respect to $v$-topology.
\item [(ii)] $T_N(\,\cdot\,,\mu) \in C_v(\sX)$ for all $\mu \in \P_v(\sX)$ and $N\geq1$.
\item [(iii)] $\sup_{N\geq1} \|T_N(\,\cdot\,,\mu)\|_v < \infty$ for all $\mu \in \P_v(\sX)$.
\item [(iv)] The function $F_t(\mu) \coloneqq \sup_{\substack{x \in \sX \\ N\geq1}} |T_N(x,\mu) - T_N(x,\mu_t)|$ on $\P_v(\sX)$ is real valued and $\|F_t\|^*_v < \infty$.
\end{itemize}
\end{lemma}

\begin{proof}
Fix any sequence $\{T_N\}_{N\geq1}$ satisfying the hypothesis of the lemma. Then, we have
\begin{align}
&\bigl| {\cal L}(\tx_1^N(t),\te_t^{(N)})(T_N) - {\cal L}(\hx_1^N(t),\delta_{\mu_t})(T_N) \bigr| \nonumber \\
&\phantom{xxxx}\leq \biggl| \int_{\sX \times \P_v(\sX)} T_N(x,\mu) {\cal L}(\tx_1^N(t),\te_t^{(N)})(dx,d\mu) \nonumber \\
&\phantom{xxxxxxxxxxxxx}- \int_{\sX \times \P_v(\sX)} T_N(x,\mu) {\cal L}(\tx_1^N(t),\delta_{\mu_t})(dx,d\mu) \biggr| \nonumber \\
&\phantom{xxxxxxxx}+ \biggl| \int_{\sX \times \P_v(\sX)} T_N(x,\mu) {\cal L}(\tx_1^N(t),\delta_{\mu_t})(dx,d\mu) \nonumber \\
&\phantom{xxxxxxxxxxxxx}- \int_{\sX \times \P_v(\sX)} T_N(x,\mu) {\cal L}(\hx_1^N(t),\delta_{\mu_t})(dx,d\mu) \biggr|. \label{eq12}
\end{align}
First, let us consider the second term in (\ref{eq12}). We have
\begin{align}
\lim_{N\rightarrow\infty} \biggl| \int_{\sX} T_N(x,\mu_t) {\cal L}(\tx_1^N(t))(dx) - \int_{\sX} T_N(x,\mu_t) {\cal L}(\hx_1^N(t))(dx) \biggr|=0 \nonumber
\end{align}
as $\{T_N(\,\cdot\,,\mu_t)\}_{N\geq1} \subset C_v(\sX)$ with $\sup_{N\geq1} \|T_N(\,\cdot\,,\mu_t)\|_v < \infty$.

Now, consider the first term in (\ref{eq12}). To this end, let us define ${\cal F} \coloneqq \bigl\{T_N(x,\,\cdot\,): x \in \sX, N\geq1)\bigr\}$. Then, we have
\begin{align}
&\lim_{N\rightarrow\infty}\biggl| \int_{\sX \times \P_v(\sX)} T_N(x,\mu) {\cal L}(\tx_1^N(t),\te_t^{(N)})(dx,d\mu) \nonumber \\
&\phantom{xxxxxxxxxxxxxxx}- \int_{\sX \times \P_v(\sX)} T_N(x,\mu) {\cal L}(\tx_1^N(t),\delta_{\mu_t})(dx,d\mu) \biggr| \nonumber \\
&\leq \lim_{N\rightarrow\infty} \int_{\sX} \biggl| \int_{\P_v(\sX)} T_N(x,\mu) {\cal L}(\te_t^{(N)}|\tx_1^N(t))(d\mu|dx)  \nonumber \\
&\phantom{xxxxxxxxxxxxxxx}- \int_{\P_v(\sX)} T_N(x,\mu) {\cal L}(\delta_{\mu_t})(d\mu) \biggr| {\cal L}(\tx_1^N(t))(dx) \nonumber \\
&\leq \lim_{N\rightarrow\infty} E\biggl[ E\biggl[ \bigl|T_N(\tx_1^N(t),\te_t^{(N)}) - T_N(\tx_1^N(t),\mu_t) \bigl| \biggl| \tx_1^N(t) \biggr] \biggr] \nonumber \\
&\leq \lim_{N\rightarrow\infty} E\biggl[ \sup_{\Psi \in {\cal F}} \bigl|\Psi(\te_t^{(N)}) - \Psi(\mu_t) \bigl| \biggr] \nonumber \\
&= 0 \text{  } (\text{by Lemma~\ref{lemma6}}) \nonumber.
\end{align}
This completes the proof.
\end{proof}

For any sequence $\{g_N\}_{N\geq1} \subset C_v(\sX)$ with $\sup_{N\geq1} \|g_N\|_v \eqqcolon L <\infty$ and for any $t\geq0$, let us define
\begin{align}
l_{N,t}(x,\mu) \coloneqq \int_{\sA \times \sX} g_N(y) p(dy|x,a,\mu) \tpi_t^{(N)}(da|x). \nonumber
\end{align}
Note that for any $(\mu,\nu) \in \P_v(\sX)^2$, we have
\begin{align}
\sup_{\substack{x \in \sX \\ N\geq1}}|l_{N,t}(x,\mu) - l_{N,t}(x,\nu)| &\leq L \sup_{x \in \sX } \int_{\sA} \|p(\,\cdot\,|x,a,\mu) - p(\,\cdot\,|x,a,\nu)\|_{v} \tpi_t^{(N)}(da|x) \nonumber \\
&\leq L \omega_p(\tilde{\rho}_v(\mu,\nu)). \nonumber
\end{align}
Since $\omega_p(r) \rightarrow0$ as $r\rightarrow0$ by Assumption~2-(h), the family $\{l_{N,t}(x,\,\cdot\,): x \in \sX, N\geq1\}$ is equicontinuous. Moreover, the function
\begin{align}
F_t(\mu) \coloneqq \sup_{\substack{x \in \sX \\ N\geq1}} |l_{N,t}(x,\mu) - l_{N,t}(x,\mu_t)| \leq L \omega_p(\tilde{\rho}_v(\mu,\mu_t)) \nonumber
\end{align}
is real-valued and $\|F_t\|^*_v<\infty$ since $\|\omega_p(\tilde{\rho}_v(\,\cdot\,,\mu_t))\|^*_v<\infty$ by again Assumption~2-(h). Note also that $l_{N,t}(\,\cdot\,,\mu) \in C_v(\sX)$ for all $\mu \in \P_v(\sX)$ and $N\geq1$, and $\sup_{N\geq1} \|l_{N,t}(\,\cdot\,,\mu)\|_v<\infty$ for all $\mu \in \P_v(\sX)$.

Indeed, we have $|l_{N,t}(x,\mu)| \leq L \alpha v(x)$, and so, $l_{N,t}(\,\cdot\,,\mu) \in B_v(\sX)$ for all $\mu \in \P_v(\sX)$. Furthermore, if $x_n \rightarrow x$ in $\sX$, then since
\begin{align}
r_n(a) &\coloneqq \int_{\sX} g_N(y) p(dy|x_n,a,\mu) \nonumber \\
\intertext{continuously converges to}
r(a) &\coloneqq \int_{\sX} g_N(y) p(dy|x,a,\mu) \nonumber
\end{align}
and $\|r_n\| \leq L v(x_n) \leq L \sup_{n\geq1} v(x_n) < \infty$, we have $\lim_{n\rightarrow\infty} l_{N,t}(x_n,\mu) = l_{N,t}(x,\mu)$ by \cite[Theorem 3.3]{Ser82} as $\tpi_t^{(N)}$ is assumed to be weakly continuous. Hence, $l_{N,t}(\,\cdot\,,\mu) \in C_v(\sX)$. We also have
\begin{align}
\sup_{N\geq1} \sup_{x \in \sX} \frac{|l_{N,t}(x,\mu)|}{v(x)} &= \sup_{N\geq1} \sup_{x \in \sX} \frac{\bigl|\int_{\sA \times \sX}g_N(y) p(dy|x,a,\mu) \tpi_t^{(N)}(da|x)\bigr|}{v(x)} \nonumber \\
&\leq L \sup_{x \in \sX} \frac{\int_{\sA \times \sX} v(y) p(dy|x,a,\mu) \tpi_t^{(N)}(da|x)}{v(x)} \nonumber \\
&\leq L \alpha <\infty. \nonumber
\end{align}
Hence, $\sup_{N\geq1} \|l_{N,t}(\,\cdot\,,\mu)\|_v<\infty$. Using these observations, we now prove the following result.

\begin{proposition}\label{prop9}
For any sequence $\{g_N\}_{N\geq1} \subset C_v(\sX)$ with $\sup_{N\geq1} \|g_N\|_v <\infty$ and for any $t\geq0$, we have
\begin{align}
\lim_{N\rightarrow\infty} \bigl| {\cal L}(\tx_1^N(t))(g_N) - {\cal L}(\hx_1^N(t))(g_N) \bigr| = 0. \nonumber
\end{align}
Therefore, by Lemma~\ref{lemma7}, for any $t\geq0$, we have
\begin{align}
\lim_{N\rightarrow\infty} \bigl| {\cal L}(\tx_1^N(t),\te_t^{(N)})(T_N) - {\cal L}(\hx_1^N(t),\delta_{\mu_t})(T_N) \bigr| = 0 \nonumber
\end{align}
for any sequence $\{T_N\}$ of functions on $\sX \times \P_v(\sX)$ satisfying hypothesis (i)-(iv) in Lemma~\ref{lemma7}.
\end{proposition}

\begin{proof}
We prove the result by induction on $t$. The claim trivially holds for $t=0$ as ${\cal L}(\tx_1^N(0)) = {\cal L}(\hx_1^N(0)) = \mu_0$ for all $N\geq1$. Suppose the claim holds for $t$ and consider $t+1$. We can write
\begin{align}
\bigl| {\cal L}(\tx_1^N(t+1))(g_N) - {\cal L}(\hx_1^N(t+1))(g_N) \bigr| &=\biggl| \int_{\sX \times \P_v(\sX)} l_{N,t}(x,\mu) {\cal L}(\tx_1^N(t),\te_t^{(N)})(dx,d\mu) \nonumber \\
&\phantom{}- \int_{\sX \times \P_v(\sX)} l_{N,t}(x,\mu) {\cal L}(\hx_1^N(t),\delta_{\mu_t})(dx,d\mu) \biggr|. \nonumber
\end{align}
Since the family $\{l_{N,t}\}_{N\geq1}$ satisfies the hypothesis of Lemma~\ref{lemma7} as we showed above, the last term converges to zero as $N\rightarrow\infty$. This completes the proof.
\end{proof}

Recall that we have
\begin{align}
J_{\bmu}(\tpi^{(N)}) = \sum_{t=0}^{\infty} \beta^t E\bigl[ c(\hx^N(t),\hat{a}^N(t),\mu_t) \bigr]. \nonumber
\end{align}

\begin{theorem}\label{theorem3}
Let $\{\tpi^{(N)}\}_{N\geq1} \subset \sM_1^c$ be an arbitrary sequence of policies for Agent~$1$. Then, we have
\begin{align}
\lim_{N \rightarrow \infty} \bigl| J_1^{(N)}(\tpi^{(N)},\pi,\ldots,\pi) - J_{\bmu}(\tpi^{(N)}) \bigr| = 0. \nonumber
\end{align}
\end{theorem}

\begin{proof}
Fix any $t\geq0$ and define
\begin{align}
T_{N,t}(x,\mu) \coloneqq \int_{\sA} c(x,a,\mu) \tpi_t^{(N)}(da|x). \nonumber
\end{align}
Note that for any $(\mu,\nu) \in \P_v(\sX)^2$, we have
\begin{align}
\sup_{\substack{x \in \sX \\ N\geq1}}\bigl| T_{N,t}(x,\mu) - T_{N,t}(x,\nu) \bigr| &\phantom{}= \sup_{\substack{x \in \sX \\ N\geq1}}\biggl| \int_{\sA} c(x,a,\mu) \tpi_t^{(N)}(da|x) - \int_{\sA} c(x,a,\nu) \tpi_t^{(N)}(da|x) \biggr| \nonumber \\
&\phantom{}\leq \omega_c(\tilde{\rho}_v(\mu,\nu)). \nonumber
\end{align}
Since $\omega_c(r) \rightarrow 0$ as $r\rightarrow0$, the family $\bigl\{T_{N,t}(x,\,\cdot\,): x \in \sX, N\geq1\bigr\}$ is equicontinuous. Moreover, the function
\begin{align}
F_t(\mu) \coloneqq \sup_{\substack{x \in \sX \\ N\geq1}}\bigl| T_{N,t}(x,\mu) - T_{N,t}(x,\mu_t) \bigr| \nonumber \end{align}
is real-valued and $\|F_t\|^*_v<\infty$.

One can also prove that $T_N(\,\cdot\,,\mu) \in C_v(\sX)$ for all $N\geq1$ and $\mu \in \P_v(\sX)$, and $\sup_{N\geq1} \|T_N(\,\cdot\,,\mu)\|_v < \infty$ for all $\mu \in \P_v(\sX)$.
Therefore, by Proposition~\ref{prop9}, we have
\begin{align}
\lim_{N\rightarrow\infty} \biggl| E\bigl[ c(\tx_1^N,\tilde{a}_1^N,\te_t^{(N)}) \bigr] - E\bigl[ c(\hx_1^N,\hat{a}_1^N,\mu_t) \bigr] \biggr|=0. \nonumber
\end{align}
Since $t$ is arbitrary, above result holds for all $t\geq0$. By using the same method as in the last part of Proposition~\ref{prop6}, we can conclude that
\begin{align}
\lim_{N \rightarrow \infty} \bigl| J_1^{(N)}(\tpi^{(N)},\pi,\ldots,\pi) - J_{\bmu}(\tpi^{(N)}) \bigr| = 0. \nonumber
\end{align}
by the dominated convergence theorem.
\end{proof}

As a corollary of Proposition~\ref{prop6} and Theorem~\ref{theorem3}, we obtain the following result.

\begin{corollary}\label{cor1}
We have
\begin{align}
\lim_{N \rightarrow \infty} J_1^{(N)}(\tpi^{(N)},\pi,\ldots,\pi)
&\geq \inf_{\pi' \in \Pi} J_{\bmu}(\pi') \nonumber \\
&= J_{\bmu}(\pi) \nonumber \\
&= \lim_{N \rightarrow \infty} J_1^{(N)}(\pi,\pi,\ldots,\pi). \nonumber
\end{align}
\end{corollary}

%

\noindent Now we are ready to prove the main result of the paper:

\begin{proof}[Proof of Theorem~\ref{theorem4}]
For sufficiently large $N$, we need to prove that
\begin{align}
J_i^{(N)}({\boldsymbol \pi}^{(N)}) &\leq \inf_{\pi^i \in \sM_i^c} J_i^{(N)}({\boldsymbol \pi}^{(N)}_{-i},\pi^i) + \varepsilon \label{eq13}
\end{align}
for each $i=1,\ldots,N$. Since the transition probabilities and the one-stage cost functions are the same for all agents, it is sufficient to prove (\ref{eq13}) for Agent~$1$ only. Given $\epsilon > 0$, for each $N\geq1$, let $\tpi^{(N)} \in \sM_1^c$ be such that
\begin{align}
J_1^{(N)} (\tpi^{(N)},\pi,\ldots,\pi) < \inf_{\pi' \in \sM_1^c} J_1^{(N)} (\pi',\pi,\ldots,\pi) + \frac{\varepsilon}{3}. \nonumber
\end{align}
Then, by Corollary~\ref{cor1}, we have
\begin{align}
\lim_{N\rightarrow\infty} J_1^{(N)} (\tpi^{(N)},\pi,\ldots,\pi) &= \lim_{N\rightarrow\infty} J_{\bmu}(\tpi^{(N)}) \nonumber \\
&\geq \inf_{\pi'} J_{\bmu}(\pi') \nonumber \\
&= J_{\bmu}(\pi) \nonumber \\
&= \lim_{N\rightarrow\infty} J_1^{(N)} (\pi,\pi,\ldots,\pi). \nonumber
\end{align}
Therefore, there exists $N(\varepsilon)$ such that
\begin{align}
\inf_{\pi' \in \sM_1^c} J_1^{(N)} (\pi',\pi,\ldots,\pi) + \varepsilon &> J_1^{(N)} (\tpi^{(N)},\pi,\ldots,\pi) + \frac{2\varepsilon}{3} \nonumber \\
&\geq J_{\bmu}(\pi) + \frac{\varepsilon}{3} \nonumber \\
&\geq J_1^{(N)} (\pi,\pi,\ldots,\pi). \nonumber
\end{align}
for all $N\geq N(\varepsilon)$.
\end{proof}

\section{Example}\label{example}

In this section, we consider an example, the additive noise model, in order to illustrate our results. In this model, the dynamics of a generic agent for $N$-agent game problem are given by
\begin{align}
x_i^N(t+1) = \frac{1}{N} \sum_{j=1}^N f(x_i^N(t),a_i^N(t),x_j^N(t)) + g(x_i^N(t),a_i^N(t)) w_i^N(t), \label{model}
\end{align}
where $x_i^N(t),x_j^N(t) \in \sX$, $a_i^N(t) \in \sA$, and $w_i^N(t) \in \sW$. Here, we assume that $\sX = \sW = \R$, $\sA \subset \R$, and the `noise' $\{w_i^N(t)\}$ is a sequence of i.i.d. random variables with standard normal distribution. Note that we can write (\ref{model}) as
\begin{align}
x_i^N(t+1) &= \int_{\sX} f(x_i^N(t),a_i^N(t),y) \phantom{i} e_t^{(N)}(dy) + g(x_i^N(t),a_i^N(t)) w_i^N(t) \nonumber \\
&= F(x_i^N(t),a_i^N(t),e_t^{(N)}) + g(x_i^N(t),a_i^N(t)) w_i^N(t), \nonumber
\end{align}
where $F(x,a,\mu) \coloneqq \int_{\sX} f(x,a,y) \mu(dy)$. The one-stage cost function of a generic agent is given by
\begin{align}
c(x_i^N(t),a_i^N(t),e_t^{(N)}) &= \frac{1}{N} \sum_{j=1}^N d(x_i^N(t),a_i^N(t),x_j^N(t)) \nonumber \\
&= \int_{\sX} d(x_i^N(t),a_i^N(t),y) \phantom{i} e_t^{(N)}(dy), \nonumber
\end{align}
for some measurable function $d: \sX \times \sA \times \sX \rightarrow [0,\infty)$.

For this model, Assumption~1 holds with $w(x) = v(x) = 1+x^2$ under the following conditions: (i) $\sA$ is compact, (ii) $g$ is continuous, and $f$ is bounded and continuous, (iii) $\sup_{a \in \sA} g^2(x,a) \leq L x^2$ for some $L>0$, (iv) $\sup_{(x,a,y) \in K \times \sA \times \sX} d(x,a,y) < \infty$ for any compact $K \subset \sX$, (v) $d(x,a,y) \leq R w(x) w(y)$ for some $R>0$, (vi) $\omega_d(r) \rightarrow 0$ as $r\rightarrow0$, where
\begin{align}
\omega_d(r) = \sup_{y \in \sX} \sup_{\substack{(x,a),(x',a'):\\ |x-x'|+|a-a'| \leq r}} |d(x,a,y) - d(x,a,y)|, \nonumber
\end{align}
(vii) $\alpha^2 \beta <1$, where $\alpha = \max\{1+\|f\|, L\}$.

Indeed, we have
\begin{align}
\int_{\sX} (1+y^2) p(dy|x,a,\mu) &= \int_{\sX} \biggl( 1+ \bigl[ F(x,a,\mu) + g(x,a) y \bigr]^2 \biggr) q(y) m(dy) \nonumber \\
&\leq 1 + \|f\|^2 + g^2(x,a) \leq \alpha (1+x^2), \nonumber
\end{align}
where $q$ is density of the standard normal distribution and $m$ is the Lebesgue measure. Hence, Assumption~1-(c) holds. In order to verify Assumption~1-(d), suppose $(x_n,a_n,\mu_n) \rightarrow (x,a,\mu)$ and let $g \in C_b(\sX)$. Then, we have
\begin{align}
\lim_{n\rightarrow\infty} \int_{\sX} g(y) p(dy|x_n,a_n,\mu_n) &= \lim_{n\rightarrow\infty} \int_{\sX} g\bigl(F(x_n,a_n,\mu_n) + g(x_n,a_n) z\bigl) q(z) m(dz) \nonumber \\
&= \int_{\sX} g\bigl(F(x,a,\mu) + g(x,a) z\bigl) q(z) m(dz) \nonumber
\end{align}
since $g$ and $F$ are continuous, where the continuity of $F$ follows from \cite[Theorem 3.3]{Ser82} and the fact that $f$ is bounded and continuous. Therefore, the transition probability $p(\,\cdot\,|x,a,\mu)$ is weakly continuous. Similarly, one can verify that the function $\int_{\sX} (1+y^2) p(dy|x,a,\mu)$ is continuous. Thus, Assumption~1-(d) holds. Note that Assumption~1-(e) holds if the initial distribution $\mu_0$ has a finite second moment. Finally, we will verify Assumption~1-(a) and (f). For (a), suppose $(x_n,a_n,\mu_n) \rightarrow (x,a,\mu)$. Then, we have
\begin{align}
|c(x_n,a_n,\mu_n)& - c(x,a,\mu)| = \biggl| \int_{\sX} d(x_n,a_n,y) \mu_n(dy) - \int_{\sX} d(x,a,y) \mu(dy) \biggr| \nonumber \\
&\leq \biggl| \int_{\sX} d(x_n,a_n,y) \mu_n(dy) - \int_{\sX} d(x,a,y) \mu_n(dy) \biggr| \nonumber \\
&\phantom{xxxxxxxxxxxxxx} \biggl| \int_{\sX} d(x,a,y) \mu_n(dy) - \int_{\sX} d(x,a,y) \mu(dy) \biggr| \nonumber \\
&\leq \omega_d\bigl(|x_n-x| + |a_n-a|\bigr) + \biggl| \int_{\sX} d(x,a,y) \mu_n(dy) - \int_{\sX} d(x,a,y) \mu(dy) \biggr|. \nonumber
\end{align}
The last expression goes to zero as $N\rightarrow\infty$, since $d(x,a,\,\cdot\,) \in C_b(\sX)$ and $\omega_d(r) \rightarrow0$ as $r\rightarrow0$. Thus, $c$ is continuous. For (f), we have
\begin{align}
\sup_{(a,\mu) \in \sA \times \P_v^t(\sX)} c(x,a,\mu) &= \sup_{(a,\mu) \in \sA \times \P_v^t(\sX)} \int_{\sX} d(x,a,y) \mu(dy) \nonumber \\
&\leq R v(x) \sup_{(a,\mu) \in \sA \times \P_v^t(\sX)} \int_{\sX} v(y) \mu(dy) \leq R \alpha^t v(x). \nonumber
\end{align}
Hence, Assumption~1-(f) holds with $\gamma = \alpha$. Therefore, under (i)-(vii), there exists a mean-field equilibrium for the mean-field game that arises from the above finite-agent game problem.

For the same model, Assumption~2 holds under the following conditions: (viii) $d(x,a,y)$ is (uniformly) H\"{o}lder continuous in $y$ with exponent $p$ and H\"{o}lder constant $K_d$, (ix) $f(x,a,y)$ is (uniformly) H\"{o}lder continuous in $y$ with exponent $p$ and H\"{o}lder constant $K_f$, (x) $g$ is bounded and $\inf_{(x,a) \in \sX \times \sA} |g(x,a)| \eqqcolon \theta > 0$.

Indeed, we have
\begin{align}
\omega_c(r) &= \sup_{(x,a) \in \sX\times\sA} \sup_{\substack{\mu,\nu: \\ \tilde{\rho}_v(\mu,\nu)\leq r}} \biggl |\int_{\sX} d(x,a,y) \mu(dy) - \int_{\sX} d(x,a,y) \nu(dy) \biggr| \nonumber \\
&\leq \sup_{(x,a) \in \sX\times\sA} \sup_{\substack{\mu,\nu: \\ \tilde{\rho}_v(\mu,\nu)\leq r}} K_d W_p(\mu,\nu)^p \nonumber \\
&\leq K_d r^p. \nonumber
\end{align}
Hence, $\omega_c(r) \rightarrow 0$ as $r\rightarrow0$ and $\omega_c(\tilde{\rho}(\,\cdot\,,\mu))$ has a finite $v$-norm. Therefore, Assumption~2-(h) holds for $\omega_c$. For $\omega_p$, we have
\begin{align}
\omega_{p}(r) &= \sup_{(x,a) \in \sX\times\sA} \sup_{\substack{\mu,\nu: \\ \tilde{\rho}_v(\mu,\nu)\leq r}} \|p(\,\cdot\,|x,a,\mu) - p(\,\cdot\,|x,a,\nu)\|_{v} \nonumber \\
&= \sup_{(x,a) \in \sX\times\sA} \sup_{\substack{\mu,\nu: \\ \tilde{\rho}_v(\mu,\nu)\leq r}}
\sup_{l \in \Lip_{\lambda}(1,\R)} \biggl| \int_{\sX} l\bigl(F(x,a,\mu)+g(x,a) z\bigr) q(z) m(dz) \nonumber \\
&\phantom{xxxxxxxxxxxxxxxxxxxxxxxxx}- \biggl| \int_{\sX} l(F(x,a,\nu)+g(x,a) z) q(z) m(dz) \biggr| \nonumber \\
&= \sup_{(x,a) \in \sX\times\sA} \sup_{\substack{\mu,\nu: \\ \tilde{\rho}_v(\mu,\nu)\leq r}}
\sup_{l \in \Lip_{\lambda}(1,\R)} \biggl| \int_{\sX} l\bigl(g(x,a) z\bigr) q\biggl(z-\frac{F(x,a,\mu)}{g(x,a)}\biggr) m(dz) \nonumber \\
&\phantom{xxxxxxxxxxxxxxxxxxxxxxxxx}- \biggl| \int_{\sX} l\bigl(g(x,a) z\bigr) q\biggl(z-\frac{F(x,a,\nu)}{g(x,a)}\biggr) m(dz) \biggr| \nonumber \\
&\leq \|l\| \sup_{(x,a) \in \sX\times\sA} \sup_{\substack{\mu,\nu: \\ \tilde{\rho}_v(\mu,\nu)\leq r}}
\sup_{h \in \Lip_{\lambda}(1,\R)} \biggl| \int_{\sX} h(z) q\biggl(z-\frac{F(x,a,\mu)}{g(x,a)}\biggr) m(dz) \nonumber \\
&\phantom{xxxxxxxxxxxxxxxxxxxxxxxxx}- \int_{\sX} h(z) q\biggl(z-\frac{F(x,a,\nu)}{g(x,a)}\biggr) m(dz) \biggr|, \label{bound}
\end{align}
where (\ref{bound}) follows from the fact that $l(g(x,a) \,\cdot\,) \in \Lip_{\lambda}(\|l\|,\R)$. Note that for any compact interval $K = [-k,k] \subset \sX$, we can upper bound (\ref{bound}) as follows:
\begin{align}
&(\ref{bound}) \leq \|l\| \sup_{(x,a) \in \sX\times\sA} \sup_{\substack{\mu,\nu: \\ \tilde{\rho}_v(\mu,\nu)\leq r}}
\sup_{h \in \Lip_{\lambda}(1,\R)} \biggl( \biggl| \int_{K} h(z) q\biggl(z-\frac{F(x,a,\mu)}{g(x,a)}\biggr) m(dz) \nonumber \\
&\phantom{xxxxxxxxxxxxxxxxxxxxxxxxxxxx}- \int_{K} h(z) q\biggl(z-\frac{F(x,a,\nu)}{g(x,a)}\biggr) m(dz) \biggr| \nonumber \\
&+ \int_{K^c} |h(z)| q\biggl(z-\frac{F(x,a,\mu)}{g(x,a)}\biggr) m(dz) + \int_{K^c} |h(z)| q\biggl(z-\frac{F(x,a,\nu)}{g(x,a)}\biggr) m(dz) \biggr) \nonumber \\
&\leq \|l\| \sup_{(x,a) \in \sX\times\sA} \sup_{\substack{\mu,\nu: \\ \tilde{\rho}_v(\mu,\nu)\leq r}}
\sup_{h \in \Lip_{\lambda}(1,\R)} \biggl( \biggl| \int_{K} h(z) q\biggl(z-\frac{F(x,a,\mu)}{g(x,a)}\biggr) m(dz) \nonumber \\
&\phantom{xxxxxxxxxxxxxxxxxxxxxxxxxxxx}- \int_{K} h(z) q\biggl(z-\frac{F(x,a,\nu)}{g(x,a)}\biggr) m(dz) \biggr| \nonumber \\
&+ \int_{K^c} (1+z^2) q\biggl(z-\frac{F(x,a,\mu)}{g(x,a)}\biggr) m(dz) + \int_{K^c} (1+z^2) q\biggl(z-\frac{F(x,a,\nu)}{g(x,a)}\biggr) m(dz) \biggr) \label{bound2}
\end{align}
Note that the last two terms in the last expression go to zero (uniformly in $(x,a,\mu,\nu)$) as $k \rightarrow \infty$, since $F$ and $g$ are bounded, and $\inf_{(x,a) \in \sX \times \sA} |g(x,a)| > 0$. For any $\varepsilon>0$, let $K_{\varepsilon}=[-k_{\varepsilon},k_{\varepsilon}] \subset \sX$ so that the sum of these terms is less than $\varepsilon$ for all $(x,a,\mu,\nu)$. Let us define the Lipschitz seminorm of $q$ on $K_{\varepsilon}$ as
\begin{align}
T_{\varepsilon} \coloneqq \sup_{\substack{(x,z) \in K_{\varepsilon}\times K_{\varepsilon}\\ x \neq z}} \frac{|q(x) - q(z)|}{|x-z|}. \nonumber
\end{align}
Then, we have
\begin{align}
(\ref{bound2}) &\leq \|l\| \sup_{(x,a) \in \sX\times\sA} \sup_{\substack{\mu,\nu: \\ \tilde{\rho}_v(\mu,\nu)\leq r}}
\sup_{h \in \Lip_{\lambda}(1,\R)} \biggl| \int_{K_{\varepsilon}} h(z) q\biggl(z-\frac{F(x,a,\mu)}{g(x,a)}\biggr) m(dz) \nonumber \\
&\phantom{xxxxxxxxxxxxxxxxxxxxxxxx}- \int_{K_{\varepsilon}} h(z) q\biggl(z-\frac{F(x,a,\nu)}{g(x,a)}\biggr) m(dz) \biggr| +\|l\| \varepsilon \nonumber \\
&\leq \|l\| T_{\varepsilon} \sup_{(x,a) \in \sX\times\sA} \sup_{\substack{\mu,\nu: \\ \tilde{\rho}_v(\mu,\nu)\leq r}}
\sup_{h \in \Lip_{\lambda}(1,\R)} \int_{K_{\varepsilon}} |h(z)|  \biggl|\frac{F(x,a,\mu)}{g(x,a)} - \frac{F(x,a,\nu)}{g(x,a)} \biggr| m(dz) \nonumber \\
&\phantom{xxxxxxxxxxxxxxxxxxxxxxxxxxxxxxxxxxxxxxxxxxxxx}+\|l\| \varepsilon \nonumber \\
&\leq \frac{\|l\| T_{\varepsilon}}{\theta} \int_{K_{\varepsilon}} (1+z^2) m(dz) \sup_{(x,a) \in \sX\times\sA} \sup_{\substack{\mu,\nu: \\ \tilde{\rho}_v(\mu,\nu)\leq r}} |F(x,a,\mu) - F(x,a,\nu)| + \|l\| \varepsilon \nonumber \\
&\leq \frac{\|l\| T_{\varepsilon}}{\theta} \int_{K_{\varepsilon}} (1+z^2) m(dz) \sup_{(x,a) \in \sX\times\sA} \sup_{\substack{\mu,\nu: \\ \tilde{\rho}_v(\mu,\nu)\leq r}} K_f W_p(\mu,\nu)^p + \|l\| \varepsilon \nonumber \\
&\leq \frac{\|l\| T_{\varepsilon}}{\theta} \int_{K_{\varepsilon}} (1+z^2) m(dz) K_f r^p + \|l\| \varepsilon. \nonumber
\end{align}
Since $\varepsilon$ is arbitrary, we have $\omega_p(r) \rightarrow 0$ as $r\rightarrow0$ and $\omega_p(\tilde{\rho}_v(\,\cdot\,,\mu))$ has finite $v$-norm. Thus, Assumption~2-(h) holds. Finally, it is straightforward to prove that Assumption~2-(i) holds, since the noise has a standard normal distribution.
Therefore, under (viii)-(x), Assumption~2 holds.

The final piece in order to deduce the existence of approximate Markov-Nash equilibria in the above game with sufficiently many agents is to prove condition (j). However, verification of (j) is highly dependent on the systems components, and so, it is quite difficult to find a global assumption in order to satisfy (j). One way to establish this is as follows. For any $h \in C_v(\sX)$ and $\nu \in \P_v(\sX)$, define
\begin{align}
R_{\nu,h}(x,a) &\coloneqq c(x,a,\nu) + \beta \int_{\sX} h(y) p(dy|x,a,\nu) \nonumber \\
&= c(x,a,\nu) + \beta \int_{\sX} h\bigl(F(x,a,\nu)+g(x,a) z\bigr) q(z) m(dz). \nonumber
\end{align}
Using this, let us state the following condition:
\begin{itemize}
\item [(j')] For any $h \in C_v(\sX)$ and $\nu \in \P_v(\sX)$, there a exists unique minimizer $a_x \in \sA$ of $R_{\nu,h}(x,\,\cdot\,)$, for each $x \in \sX$.
\end{itemize}
Under assumption (j'), we can establish that the policy in the mean field equilibrium satisfies the weak continuity condition (j).

Indeed, fix any $t\geq0$ and consider the policy $\pi_t$ at time $t$ in $\pi$, where $\pi$ is the policy in the mean-field equilibrium. By (j'), we must have $\pi_t(\,\cdot\,|x) = \delta_{f_t(x)}(\,\cdot\,)$ for some $f_t: \sX \rightarrow \sA$ which minimizes some function $R_{\nu,h}(x,a)$ of the above form; that is, $\min_{a \in \sA} R_{\nu,h}(x,a) = R_{\nu,h}(x,f_t(x))$ for all $x \in \sX$. If $f_t$ is continuous, then $\pi_t$ is also continuous. Hence, in order to prove (j), it is sufficient to prove that $f_t$ is continuous. Suppose $x_n \rightarrow x$ in $\sX$. Note that $l(\,\cdot\,) = \min_{a \in \sX} R_{\nu}(\,\cdot\,,a)$ is continuous. Therefore, every accumulation point of the sequence $\{f_t(x_n)\}_{n\geq1}$ must be a minimizer for $R_{\nu,h}(x,\,\cdot\,)$. Since there exists a unique minimizer $f_t(x)$ of $R_{\nu,h}(x,\,\cdot\,)$, the set of all accumulation points of $\{f_t(x_n)\}_{n\geq1}$ must be $\{f_t(x)\}$. This implies that $f_t(x_n)$ converges to $f_t(x)$ since $\sA$ is compact. Hence, $f_t$ is continuous.


\begin{remark}
We note that, using similar ideas, the finite-horizon cost criterion; that is,
\begin{align}
E \biggl[ \sum_{t=0}^{T} c(x(t),a(t),\mu_t)\biggr] \text{ } \text{ for some $T<\infty$},
\end{align}
can be handled with the same quantitative results. The only part that needs to verified differently from the infinite-horizon case is Proposition~\ref{prop4}. Namely, we have to establish that $F_t^{(n)}$ and $J_{*,t}^{\bnu^{(n)}}$ continuously converge to $F_t$ and $J^{\bnu}_{*,t}$, respectively, since we cannot use the fixed-point argument anymore (see the proof of Proposition~\ref{prop4}). Note that, in finite-horizon case, for each $n$ and $t<T$, these functions are given by
\begin{align}
F^{(n)}_t(x,a) &= c(x,a,\nu^{(n)}_{t,1}) + \int_{\sX} J^{\bnu^{(n)}}_{*,t+1}(y) p(dy|x,a,\nu^{(n)}_{t,1}), \nonumber \\
F_t(x,a) &= c(x,a,\nu_{t,1}) + \int_{\sX} J^{\bnu}_{*,t+1}(y) p(dy|x,a,\nu_{t,1}), \nonumber
\end{align}
and
\begin{align}
J^{\bnu^{(n)}}_{*,t}(x) = \min_{a \in \sA} F^{(n)}_t(x,a) \text{ } \text{ and } \text{ } J^{\bnu}_{*,t}(x) = \min_{a \in \sA} F_t(x,a). \nonumber
\end{align}
Observe that the discount factor $\beta$ is missing in the above equations. For $t=T$, we have $F_T^{(n)}(x,a) = c(x,a,\nu_{T,1}^{(n)})$ and $F_T(x,a) = c(x,a,\nu_{T,1})$. Since $c$ is continuous and $\nu_{T,1}^{(n)}$ weakly converges to $\nu_{T,1}$, we have that $J^{\bnu^{(n)}}_{*,T}$ continuously converges to $J^{\bnu}_{*,T}$ by \cite[Proposition 7.32]{BeSh78}. But this implies that $F_{T-1}^{(n)}$ continuously converges to $F_{T-1}$ by the discussion before Proposition~\ref{prop4}, and so, $J^{\bnu^{(n)}}_{*,T-1}$ continuously converges to $J^{\bnu}_{*,T-1}$ again by \cite[Proposition 7.32]{BeSh78}.
Then, by induction hypothesis, we can conclude that $F_t^{(n)}$ and $J_{*,t}^{\bnu^{(n)}}$ continuously converge to $F_t$ and $J^{\bnu}_{*,t}$, respectively, for each $t\leq T$. Therefore, Theorems~\ref{thm:MFE} and \ref{theorem4} hold for the finite-horizon cost criterion under the same set of assumptions. Furthermore, if we start the mean-field game at time $\tau>0$ with initial measure $\mu_{\tau}$, then the pair $\bigl(\{\pi_t\}_{\tau\leq t \leq T},\{\mu_t\}_{\tau\leq t \leq T}\bigr)$ in Theorem~\ref{thm:MFE} is still a mean-field equilibrium for the sub-game. Similar conclusion can be reached for the finite-agent game problem; that is, the policy $\{\pi_t\}_{\tau\leq t \leq T}$ is an $\varepsilon$-Markov-Nash equilibrium for the finite-agent game problem starting at time $\tau$ with i.i.d. initial measures drawn according to $\mu_{\tau}$.

\end{remark}

\section{Conclusion}\label{conc}

Using the mean-field approach, we have shown that there exist nearly Markov-Nash equilibria for finite-population
game problems with infinite-horizon discounted costs when the number of agents is sufficiently large. Under mild technical conditions, we have first established the existence of a Nash equilibrium in the limiting mean-field game problem. We have then applied this policy to the finite population game and have demonstrated that it constitutes an $\varepsilon$-Markov-Nash equilibrium for games with a sufficiently large number of agents.


One interesting future direction is to study mean-field games with imperfect information. In this case, one possible approach is to use the theory developed for partially observed Markov decision processes (POMDPs). Finally, average-cost and risk-sensitive optimality criteria are also worth studying. In particular, using the vanishing discount approach in MDP theory (i.e., $\beta \rightarrow 1$), we may be able to establish similar results for the average cost case.

\section*{Appendix}

\appsec

\subsection{Proof of Theorem~\ref{theorem0}}\label{app0}

Since the transition probabilities and the one-stage cost functions are the same for all agents, it is sufficient to prove the result for Agent~$1$. To this end, choose an arbitrary policy ${\boldsymbol \pi}^{(N)} \subset {\bf \sM}^{(N)}$. We prove that for any $\varepsilon>0$, there exists $\pi^{1,\varepsilon} \in \sM_1^c$ such that
\begin{align}
J_1^{(N)}(\pi^{1,\varepsilon},\pi^2,\ldots,\pi^N) < J_1^{(N)}(\pi^{1},\pi^2,\ldots,\pi^N) + \varepsilon, \nonumber
\end{align}
which will complete the proof.

To achieve this, we first describe the overall $N$-agent game as an Markov Decision Process (MDP) with state space $\sX^N \times \P(\sX)$, action space $\sA^N$, and transition probability
\begin{align}
Q(d\by,d\mu|\bx,\nu,\ba) \coloneqq \delta_{F_N(\by)}(d\mu) \prod_{i=1}^N p(dy_i|x_i,a_i,\nu), \nonumber
\end{align}
where $F_N(\by) \coloneqq \frac{1}{N} \sum_{i=1}^N \delta_{y_i}$. The one-stage cost function is given by $C(\bx,\mu,\ba) = c(x_1,a_1,\mu)$ which corresponds to the one-stage cost function of Agent~$1$. The cost function for this MDP is $\beta$-discounted cost and denoted by $J(\pi^1,\ldots,\pi^N)$ for the policy ${\boldsymbol \pi}^{(N)} = (\pi^1,\ldots,\pi^N)$. Observe that for any $\tpi \in \sM$, we have
\begin{align}
J(\tpi,\pi^2,\ldots,\pi^N) = J_1^{(N)}(\tpi,\pi^2,\ldots,\pi^N). \nonumber
\end{align}
Let $(\bx(t),\nu_t)$ and $\ba(t)$ denote the state and action of this MDP at time $t$, respectively. Let $\bx(t) = (x_1(t),\ldots,x_N(t))$.

Recall that $\mu_0 = {\cal L}(x_1(0))$. By Lusin's theorem \cite[Theorem 7.5.2]{Dud89}, for any $\delta_0>0$, there exists a closed set $F_0 \subset \sX$ such that $\mu_0(F_0) < \delta_0$ and $\pi^1_0$ is weakly continuous on $F_0$. Since $\P(\sA)$ is a convex subset of a locally convex vector space (i.e., the set of finite signed measures over $\sA$), by Dugundji extension theorem \cite[Theorem 7.4]{GrDu03} we can extend $\pi^1_{0}: F_0\rightarrow \P(\sA)$ to $\sX$ continuously. Let $\pi^1_{0,\delta}$ denote this extended continuous function and define $\pi^{1,1}_{\delta} \coloneqq (\pi^1_{0,\delta},\pi^1_1,\pi^1_2,\ldots)$. Then, it is straightforward to prove that \begin{align}
&|J(\pi^{1,1}_{\delta},\pi^2,\ldots,\pi^N) - J(\pi^{1},\pi^2,\ldots,\pi^N)| \nonumber \\
&\phantom{xxxxxxxxxxxxxx}\biggl(= |J_1^{(N)}(\pi^{1,1},\pi^2,\ldots,\pi^N) - J_1^{(N)}(\pi^{1},\pi^2,\ldots,\pi^N)| \biggr) \nonumber \\
&\phantom{xxxxxxxxxxxxxx}\leq 2 L_0 \int_{F_0^c} v(x_0) \mu_0(dx_0). \nonumber
\end{align}
Note that the last expression can be made arbitrarily small by choosing arbitrarily small $\delta_0$ since $v$ is $\mu_0$-integrable.

We can apply the same method to the policy $\pi^{1,1}_{\delta}$; that is, we replace $\pi^1_1$ with a continuous $\pi^1_{1,\delta}$ that agrees with $\pi^1_1$ on some closed subset $F_1$ of $\sX$ with probability $1-\delta_1$. Let $\pi^{1,2}_{\delta} \coloneqq (\pi^1_{0,\delta},\pi^1_{1,\delta},\pi^1_2,\ldots)$. Therefore, $|J(\pi^{1,2}_{\delta},\pi^2,\ldots,\pi^N) - J(\pi^{1,1}_{\delta},\pi^2,\ldots,\pi^N)| \leq 2 L_1 \int_{F_1^c} v(x_1) \mu_1(dx_1)$, where $\mu_1 = {\cal L}(x_1(1))$ under $\pi^{1,1}_{\delta}$ or $\pi^{1,2}_{\delta}$. Continuing in this way, we will obtain $\pi^{1,\infty}_{\delta} \in \sM_1^c$ such that
\begin{align}
&|J(\pi^{1},\pi^2,\ldots,\pi^N) - J(\pi^{1,\infty},\pi^2,\ldots,\pi^N)| \nonumber \\
&\phantom{xxxxxxxxxxxxx}\biggl(=|J_1^{(N)}(\pi^{1},\pi^2,\ldots,\pi^N) - J_1^{(N)}(\pi^{1,\infty},\pi^2,\ldots,\pi^N)| \biggr) \nonumber \\
&\phantom{xxxxxxxxxxxxx}\leq \sum_{t=0}^{\infty} |J(\pi^{1,t},\pi^2,\ldots,\pi^N) - J(\pi^{1,t+1},\pi^2,\ldots,\pi^N)| \nonumber \\
&\phantom{xxxxxxxxxxxxx}\leq \sum_{t=0}^{\infty} 2 L_{t} \beta^t \int_{F_t^c} v(x_t) \mu_t(dx_t). \nonumber
\end{align}
The last expression goes to zero with the proper choice of the sequence $\{\delta_t\}_{t\geq0}$.
This completes the proof.

\subsection{Proof of Proposition~\ref{lemma1}}\label{app01}

Let $\bmu \in \M$ and $\pi \in \Pi$ be arbitrary. According to the Ionescu-Tulcea theorem, an initial distribution $\mu_0$ on $\sX$, a policy $\pi$, and state measure flow $\bmu$ define a unique probability measure $\rP^{\pi}$ on $\sG_{\infty}=(\sX\times\sA)^{\infty}$. Define $\mu_t^{\pi} \coloneqq \rP^{\pi}(dx(t))$ for each $t\geq0$.

Let $\hat{\pi}$ be the Markov policy given by
\begin{align}
\hat{\pi}_t(da(t)|x(t)) \coloneqq \rP^{\pi}(da(t)|x(t)). \nonumber
\end{align}
We first prove that $\rP^{\hat{\pi}}(dx(t)) \eqqcolon \mu_t^{\hat{\pi}} = \mu_t^{\pi}$ for all $t$. We prove this by induction. For $t=0$ the claim clearly holds as the initial distribution $\mu_0$ is fixed. Assume the claim is true for $t\geq0$. Then we have
\begin{align}
\mu_{t+1}^{\hat{\pi}}(\,\cdot\,) &= \int_{\sX \times \sA} p(\,\cdot\,|x(t),a(t),\mu_t) \hat{\pi}_t(da(t)|x(t)) \mu_t^{\hat{\pi}}(dx(t)) \nonumber \\
&= \int_{\sX \times \sA} p(\,\cdot\,|x(t),a(t),\mu_t) \rP^{\pi}(da(t)|x(t)) \mu_t^{\pi}(dx(t)) \nonumber \\
&= \mu_{t+1}^{\pi}(\,\cdot\,). \nonumber
\end{align}
This proves the claim.

To complete the proof of the first statement, take any $t\geq0$. Then we have
\begin{align}
E^{\hat{\pi}} \bigl[ c(x(t),a(t),\mu_t^{\hat{\pi}}) \bigr] &= \int_{\sX \times \sA} c(x(t),a(t),\mu_t^{\hat{\pi}}) \hat{\pi}(da(t)|x(t)) \mu_{t}^{\hat{\pi}}(dx(t)) \nonumber \\
&= \int_{\sX \times \sA} c(x(t),a(t),\mu_t^{\pi}) \rP^{\pi}(da(t)|x(t)) \mu_{t}^{\pi}(dx(t)) \nonumber \\
&= E^{\pi} \bigl[ c(x(t),a(t),\mu_t^{\pi}) \bigr]. \nonumber
\end{align}
Since $t$ is arbitrary, the above equality holds for all $t\geq0$. Thus, by Tonelli's theorem \cite[Theorem 18.3]{Bil95}, $J_{\bmu}(\hat{\pi}) = J_{\bmu}(\pi)$.

The proof of the second statement can be done similarly, so we omit the details.

\subsection{Proof of Proposition~\ref{prop4}}\label{app1}

We use successive approximations to prove Proposition~\ref{prop4}. To ease the notation, let $T^{(n)} \coloneqq T^{\bnu^{(n)}}$, $J^{(n)}_* \coloneqq J^{\bnu^{(n)}}_*$, $T \coloneqq T^{\bnu}$, and $J_* \coloneqq J^{\bnu}_*$. Let $\bu^{(n)}_0 = \bu_0 = 0$ and
\begin{align}
\bu^{(n)}_{k+1} = T^{(n)} \bu^{(n)}_k \text{ } \text{and} \text{ } \bu_{k+1} = T \bu_k. \nonumber
\end{align}
Since $T^{(n)}$ and $T$ are contractive operators on $\C$ with modulus $\sigma \beta \alpha$, it can be proved that
\begin{align}
\rho(\bu^{(n)}_k,J^{(n)}_*) , \rho(\bu_k,J_*)  \leq (\sigma \beta \alpha)^k L_0. \label{eq7}
\end{align}

\begin{lemma}\label{lemma3}
For any compact $K \subset \sX$, we have
\begin{align}
\lim_{n\rightarrow\infty} \sup_{x \in K} \bigl| u^{(n)}_{k,t}(x) - u_{k,t}(x) \bigr| = 0, \nonumber
\end{align}
for all $k\geq0$ and $t\geq0$, where $\bu^{(n)}_k = \bigl( u^{(n)}_{k,t} \bigr)_{t\geq0}$ and $\bu_k = \bigl( u_{k,t} \bigr)_{t\geq0}$.
\end{lemma}

\begin{proof}
Fix any compact $K \subset \sX$. Since $\bu^{(n)}_0 = \bu_0 = 0$, the claim trivially holds for $k=0$ and for all $t\geq0$. Suppose that the claim holds for $k$ and for all $t\geq0$. Then, consider $k+1$ and arbitrary $t$. For these indices, we have
\begin{align}
&\sup_{x \in K} \bigl| u^{(n)}_{k+1,t}(x) - u_{k+1,t}(x) \bigr| \nonumber \\
&\phantom{xxxxxx}=\sup_{x \in K} \biggl| \min_{a \in \sA} \biggl[ c(x,a,\nu^{(n)}_{t,1}) + \beta \int_{\sX} u^{(n)}_{k,t+1}(y) p(dy|x,a,\nu^{(n)}_{t,1}) \biggr] \nonumber \\
&\phantom{xxxxxxxxxxxxxxxx}- \min_{a \in \sA}\biggl[ c(x,a,\nu_{t,1}) + \beta \int_{\sX} u_{k,t+1}(y) p(dy|x,a,\nu_{t,1}) \biggr] \biggr| \nonumber \\
&\phantom{xxxxxx}\leq \sup_{(x,a) \in K \times \sA} \bigl| c(x,a,\nu_{t,1}^{(n)}) - c(x,a,\nu_{t,1}) \bigr| \nonumber \\
&\phantom{xxxxxxxxxxxxxxxx}+ \beta \sup_{(x,a) \in K \times \sA} \biggl| \int_{\sX} u^{(n)}_{k,t+1}(y) p(dy|x,a,\nu^{(n)}_{t,1}) \nonumber \\
&\phantom{xxxxxxxxxxxxxxxxxxxxxxxxxxxxxxxxx}- \int_{\sX} u_{k,t+1}(y) p(dy|x,a,\nu_{t,1}) \biggr|. \nonumber
\end{align}
Note that $c(\cdot\,,\,\cdot\,,\nu^{(n)})$ converges to $c(\cdot\,,\,\cdot\,,\nu)$ continuously as $n\rightarrow\infty$. Furthermore,  since $u^{(n)}_{k,t+1}$ converges to $u_{k,t+1}$ continuously (as $u^{(n)}_{k,t+1}$, $u_{k,t+1}$ are continuous and $u^{(n)}_{k,t+1}$ uniformly converges to $u_{k,t+1}$ over compact sets by assumption), $u^{(n)}_{k,t+1}(y) \leq L_{t+1} v(y)$ for all $n\geq1$, and
$\int_{\sX} v(y) p(dy|x^{(n)},a^{(n)},\nu^{(n)}_{t,1})\rightarrow\int_{\sX} v(y) p(dy|x,a,\nu^{(n)}_{t,1})$ for any $(x^{(n)},a^{(n)}) \rightarrow (x,a)$ in $\sX \times \sA$, by \cite[Theorem 3.3]{Ser82} $\int_{\sX} u^{(n)}_{k,t+1}(y) p(dy|\cdot\,,\,\cdot\,,\nu^{(n)}_{t,1})$ converges to $\int_{\sX} u_{k,t+1}(y) p(dy|\cdot\,,\,\cdot\,,\nu_{t,1})$ continuously as $n\rightarrow\infty$. Since continuous convergence is equivalent to uniform convergence over compact sets for continuous functions, the last expression goes to zero as $n\rightarrow\infty$. This completes the proof.
\end{proof}

Using previous lemma, we now complete the proof of Proposition~\ref{prop4}.

Fix any compact $K \subset \sX$. For all $k\geq0$, we have
\begin{align}
\sup_{x \in K} \frac{\bigl| J^{(n)}_{*,t}(x) - J_{*,t}(x)|}{v(x)} &\leq \| J^{(n)}_{*,t} - u^{(n)}_{k,t}\|_v + \sup_{x \in K} \frac{\bigl| u^{(n)}_{k,t}(x) - u_{k,t}(x) \bigr|}{v(x)} + \|u_{k,t} - J_{*,t}\|_v \nonumber \\
&\leq \sigma^t \rho(J^{(n)}_{*},\bu^{(n)}_{k}) + \sup_{x \in K} \frac{\bigl| u^{(n)}_{k,t}(x) - u_{k,t}(x) \bigr|}{v(x)} + \sigma^t \rho(J_{*},\bu_{k}) \nonumber \\
&\leq 2 \sigma^t (\sigma \beta \alpha)^k L_0 + \sup_{x \in K} \frac{\bigl| u^{(n)}_{k,t}(x) - u_{k,t}(x) \bigr|}{v(x)} \text{ } \text{ (by (\ref{eq7}))}. \nonumber
\end{align}
This last expression can be made arbitrary small by first choosing large enough $k$ and then large enough $n$. This completes the proof since $\sup_{x \in K} v(x) < \infty$.

\subsection{Auxiliary Lemma}\label{newapp1}

In this section, we state and prove an auxiliary lemma which is a generalization of \cite[Lemma A.2]{BuMa14} to unbounded real-valued functions.

\begin{lemma}\label{auxlemma}
Let $P:\sX \rightarrow \P(\sX)$ be a transition probability on $\sX$ given $\sX$. Fix $N\geq1$ and let $y_1,\ldots,y_N \in \sX$. Let $X_1,\ldots,X_N$ be independent random variables such that ${\cal L}(X_i) = P(\,\cdot\,|y_i)$. Let $e_0^{(N)}(\,\cdot\,) \coloneqq \frac{1}{N} \sum_{i=1}^N \delta_{y_i}$ and $e_1^{(N)}(\,\cdot\,) \coloneqq \frac{1}{N} \sum_{i=1}^N \delta_{X_i}$. Then, we have
\begin{align}
E\bigl[ |e_1^{(N)}(g) - e_0^{(N)}P(g)| \bigr]^2 \leq \frac{1}{N^2} \sum_{i=1}^N \bigl( E[Z_i^2] + E[Z_i]^2 \bigr), \nonumber
\end{align}
where $Z_i = g(X_i)$.
\end{lemma}

\begin{proof}
By Jensen's inequality, we have
\begin{align}
E\bigl[ |e_1^{(N)}(g) - e_0^{(N)}P(g)| \bigr]^2 &= E \bigg[ \biggl| \frac{1}{N} \sum_{i=1}^N g(X_i) - \frac{1}{N} \sum_{i=1}^N \int_{\sX} g(y) P(dy|y_i) \biggr| \biggr]^2 \nonumber \\ &\leq E \bigg[ \biggl| \frac{1}{N} \sum_{i=1}^N g(X_i) - \frac{1}{N} \sum_{i=1}^N \int_{\sX} g(y) P(dy|y_i) \biggr|^2 \biggr]. \nonumber
\end{align}
The result follows by expanding and simplifying the term inside the expectation.
\end{proof}

\bibliographystyle{siam}
\bibliography{references}

\end{document}